\documentclass[aps,pra,reprint,showpacs,superscriptaddress,nofootinbib]{revtex4-2}
\usepackage{xr}
\usepackage[utf8]{inputenc}

\usepackage[T1]{fontenc}
\usepackage{amsmath, amsthm, amssymb, mathtools, dsfont} 
\usepackage{fixmath}
\usepackage{graphicx}  
\usepackage{dcolumn} 
\usepackage{bm,bbm} 
\usepackage[normalem]{ulem}
\usepackage[usenames,dvipsnames]{xcolor}
\usepackage{esint}
\usepackage{paralist}
\usepackage[normalem]{ulem}
\usepackage{float}
\usepackage{siunitx}
\usepackage{flushend}
\usepackage{balance}
\usepackage{soul}
\usepackage{flushend}

\usepackage[ruled]{algorithm2e}
\SetKwComment{Comment}{$\triangleright$\ }{}
\usepackage{geometry}
\definecolor{myurlcolor}{rgb}{0,0,0.7}
\definecolor{myrefcolor}{rgb}{0.8,0,0}

\usepackage[unicode=true,pdfusetitle,  bookmarks=true,bookmarksnumbered=false,bookmarksopen=false, breaklinks=false,pdfborder={0 0 0},backref=false,colorlinks=true, linkcolor=myrefcolor,citecolor=myurlcolor,urlcolor=myurlcolor]{hyperref}
\DeclareGraphicsExtensions{.png,.pdf,.eps}
\usepackage{epstopdf}

\usepackage{mathtools}

\makeatletter

 \theoremstyle{plain}
 
 \theoremstyle{plain}
 \newtheorem{lem}{Lemma}
 \theoremstyle{plain}
 
 \theoremstyle{plain}
  \newtheorem{prop}{Proposition}
 \theoremstyle{plain}
 
 \theoremstyle{plain}
 
 \theoremstyle{plain}
 
 \theoremstyle{remark}
 \newtheorem*{rem*}{Remark}
 \theoremstyle{plain}

\theoremstyle{plain}
 \newtheorem*{conj*}{Conjecture}
 \theoremstyle{plain}
  \newtheorem*{def*}{Definition}
 \theoremstyle{plain}

\newcommand{\e}{\mathrm{e}}

\renewcommand{\exp}{\mathrm{exp}}
\DeclareMathOperator{\tr}{tr}

\renewcommand{\H}{\mathcal{H}}

\newcommand{\C}{\mathbb{C}} 
\newcommand{\M}{\mathbf{M}} 
 
\newcommand{\ev}[1]{\mathbb{E}\left[ #1 \right]} 

\newcommand{\I}{\mathbb{I}} 

\definecolor{darkgreen}{rgb}{0.3, 0.8, 0.0}

\newcommand{\rev}[1]{{ #1}}

\global\long\global\long\global\long\def\bra#1{\mbox{\ensuremath{\langle#1|}}}
\global\long\global\long\global\long\def\ket#1{\mbox{\ensuremath{|#1\rangle}}}

\makeatother
\renewcommand{\ket}[1]{\left| #1 \right>} 
\renewcommand{\bra}[1]{\left< #1 \right|} 

\begin{document} 	
\title{A Simple and  Efficient Joint Measurement Strategy for\\ Estimating Fermionic Observables and Hamiltonians}

\author{Joanna Majsak}
\email{joanna.e.majsak@gmail.com}

\affiliation{Center for Theoretical Physics, Polish Academy of Sciences, Al. Lotnik\'ow 32/46, 02-668 Warsaw, Poland}
\affiliation{Quantum Research Center, Technology Innovation Institute, Abu Dhabi, UAE}
\affiliation{Faculty of Physics, University of Warsaw, Pasteura 5, 02-093 Warsaw, Poland}

\author{Daniel McNulty}
\affiliation{Center for Theoretical Physics, Polish Academy of Sciences, Al. Lotnik\'ow 32/46, 02-668 Warsaw, Poland}
\affiliation{Dipartimento di Fisica, Università di Bari, I-70126 Bari, Italy}

\author{Micha\l\ Oszmaniec}
\affiliation{Center for Theoretical Physics, Polish Academy of Sciences, Al. Lotnik\'ow 32/46, 02-668 Warsaw, Poland}
\affiliation{NASK National Research Institute, Kolska 12, 01-045 Warsaw, Poland}

\begin{abstract}
We propose a simple scheme to estimate fermionic observables and Hamiltonians relevant in quantum chemistry and correlated fermionic systems. Our approach is based on implementing a measurement that jointly measures noisy versions of any product of two or four Majorana operators in an $N$ mode fermionic system. To realize our measurement we use: (i) a randomization over a set of unitaries that realize products of Majorana fermion operators; (ii) a unitary, sampled at random from a constant-size set of suitably chosen fermionic Gaussian unitaries; (iii) a measurement of fermionic occupation numbers; (iv) suitable post-processing.  Our scheme can estimate expectation values of all quadratic and quartic Majorana monomials to $\epsilon$ precision using $\mathcal{O}(N \log(N)/\epsilon^2)$ and $\mathcal{O}(N^2 \log(N)/\epsilon^2)$ measurement rounds respectively, matching the performance offered by fermionic \rev{classical shadows} \cite{Zhao21,Wan22}. In certain settings, such as a rectangular lattice of qubits which encode an $N$ mode fermionic system via the Jordan-Wigner transformation, our scheme can be implemented in circuit depth $\mathcal{O}(N^{1/2})$ with $\mathcal{O}(N^{3/2})$ two-qubit gates, offering an improvement over fermionic and matchgate classical shadows that require depth $\mathcal{O}(N)$ and $\mathcal{O}(N^2)$ two-qubit gates. By benchmarking our method on exemplary molecular Hamiltonians and observing performances comparable to fermionic classical shadows, we demonstrate a novel, competitive alternative to existing strategies.
\end{abstract}

\maketitle	

\section{Introduction}
 Some of the most promising tasks in which quantum computers can offer practical speedups over their classical counterparts can be found in the fields of quantum chemistry and simulation (see e.g. \cite{bauer20} for a review), where fermionic systems are of particular interest. While the ideal approach would involve simulating the quantum system on a (yet to be realized) universal quantum computer, a reasonable goal in the meantime is to find problems where smaller sized quantum devices can outperform current classical approaches, given the limitations on achievable gate depths and counts. To this end, much effort has been spent developing algorithms that combine both quantum and classical techniques to achieve speedups \cite{bharti22,preskill18,mcclean16,garcia21}. These can target, e.g., the simulation of time dynamics in quantum chemistry systems \cite{miessen_quantum_2023, babbush_quantum_2023}, as well as the problem of determining low energy states of molecular Hamiltonians. A large class of these algorithms, known as variational quantum algorithms \cite{mcclean16},  require an approximation of the ground state energy of a Hamiltonian, which involves estimating the Hamiltonian on a quantum device and optimizing the prepared state by classical methods.

As the system size increases, one of the main performance bottlenecks in the aforementioned tasks is the inability to simultaneously measure the relevant non-commuting observables, e.g. those which appear in the Hamiltonian's decomposition. Non-commutativity, more generally, hinders the ability to efficiently learn many properties of a quantum system \cite{gidofalvi07,obrien19,overy14,mcclean17,takeshita20,gluza18}. To reduce the measurement cost (and the number of state preparations) several strategies have been developed. A significant body of work attempts to efficiently find minimal groupings of the measurements into sub-classes of commuting observables \cite{jena19,gokhale19,yen20,verteletskyi20,crawford21,izmaylov20,Zhao20,bonet20}. 
Another approach, known as classical shadows \cite{huang20}, is based on a randomized measurement strategy that yields a classical approximation of the quantum state. The resulting approximation can be used to simultaneously estimate non-commuting observables in both many-body qubit \cite{huang20,hadfield20,hu23} and fermionic systems \cite{Zhao21,Wan22,Low22,ogorman22}, with state-of-the-art sample complexity scalings.
A relevant approach was also recently introduced in \cite{clinton_towards_2024}, where the authors provide a method to measure all quadratic and quartic fermionic  terms in the Jordan-Wigner encoding via quantum circuits with constant depth, assuming all-to-all qubit connectivity and noiseless multi-qubit projective measurements, with sample complexity comparable to fermionic classical shadows.

In this work we consider an alternative approach for estimating non-commuting fermionic observables by implementing a simple and efficient joint measurement, in analogy to the strategy for multi-qubit Pauli observables \cite{mcnulty22}. In particular, we construct a joint measurement of a modified (noisy) collection of Majorana monomials, whose outcome statistics can be reproduced from the (parent) measurement together with an efficient classical post-processing \cite{busch,heinosaari08,heinosaari15}. Our measurement scheme for an $N$ mode fermionic system involves sampling from two distinct subsets of unitaries followed by a measurement of the occupation numbers and a simple classical post-processing. The first subset of unitaries realizes products of Majorana fermion operators. The second subset (consisting of several fermionic Gaussian unitaries) rotates disjoint blocks of the Majorana operators into balanced superpositions. For the case of quadratic and quartic monomials we show that sets of two or nine fermionic Gaussian unitaries are sufficient, respectively, to jointly measure all noisy versions of the desired observables. We also describe a tailored joint measurement scheme for estimating energies of electronic structure Hamiltonians, where four fermionic Gaussian unitaries in the second subset are sufficient. As with classical shadows, the information from a single experiment in our scheme can be used to estimate multiple observables.

Our joint measurement strategy estimates expectation values of quadratic and quartic Majorana monomials to $\epsilon$ precision with $\mathcal{O}(N \log(N)/\epsilon^2)$ and $\mathcal{O}(N^2 \log(N)/\epsilon^2)$ measurement rounds, respectively, and provides the same performance guarantees as fermionic (matchgate) classical shadows \cite{Zhao21,Wan22}. Under the Jordan-Wigner transformation, in the setting of a rectangular lattice of qubits, the measurement circuit can be achieved with depth $\mathcal{O}(N^{1/2})$ and using $\mathcal{O}(N^{3/2})$ two-qubit gates. Fermionic and matchgate classical shadow circuits, on the other hand, require depth $\mathcal{O}(N)$ and $\mathcal{O}(N^2)$ two-qubit gates, as well as requiring a randomization over a large collection of unitaries \cite{Zhao21,Wan22}. Under standard fermion-to-qubit mappings, we estimate the expectation values of Majorana pairs and quadruples from single-qubit measurement outcomes of one and two qubits, respectively. Thus each estimate is affected only by errors on at most two qubits. As such, our strategy has the potential to be easily combined with randomized error mitigation techniques (see e.g. \cite{berg_2022}). We benchmark our strategy with several popular molecular Hamiltonians and find that the sample complexities are comparable to other state-of-the-art schemes \cite{huang20,Zhao21,Huggins21,hadfield20}.

\rev{We analyze the 2D rectangular lattice as both a simple theoretical model and with practical applications in mind. Many present-day superconducting quantum computing architectures are based on 2D lattices that can be easily mapped to rectangular ones. For example, IBM's hexagon-heavy lattices can simulate a rectangular lattice with constant overhead in gate depth and count \cite{hetenyi_creating_2024}, and a rectangular lattice can be embedded into the topology of the Google Sycamore processor}.

This work is accompanied by a complementary paper \cite{MCO2024}, which studies joint measurability of Majorana fermion operators from a more formal perspective, in particular by showing (asymptotic) optimality of our joint measurement scheme and establishing connections to graph theory, as well as topics in mathematical physics such as the SYK model.    

The rest of the paper is organized as follows. First, in Sec. \ref{sec:Setting} we introduce the necessary concepts regarding fermionic systems and quantum measurements. In Sec. \ref{sec:Genprotocol} we present the general protocol for estimating expectation values of products of Majorana operators. In Sec. \ref{sec:constructionOFunitaries} we provide the general structure of the fermionic Gaussian unitaries that guarantee a competitive scaling of the variance of the estimators for Majorana pairs and quadruples. In Sec. \ref{sec:physical_hamiltonians} we specialize our scheme to quantum chemistry Hamiltonians and present its optimized implementation for a 2D qubit layout. In Sec. \ref{sec:numerical_benchmarks} we present the results of numerical benchmarks of our scheme. Finally, in Sec. \ref{sec:discussion} we conclude the paper with a short discussion and list of open problems. In Sec. \ref{sec:materials_methods} we summarize the methods used to obtain our results. In particular, in Sec. \ref{sec:proof_joint_measurement} we provide a proof justifying our proposed general protocol. In Sec. \ref{sec:realization} we discuss the gate depth and gate count required to implement our scheme on a quantum computer under different fermion-to-qubit mappings and in different qubit layouts.

\section{results}
\subsection{Setting and notation}\label{sec:Setting}

We will consider $N$ mode fermionic systems described on a Fock space $\H_N=\mathrm{Fock}_f(\C^N)$ which is a $2^N$ dimensional space spanned by Fock states $\ket{\bm{n}}=\ket{n_1,\ldots,n_N}$, where $n_i\in\lbrace{0,1\rbrace}$ are occupation numbers of fermionic particles in $N$ fermionic modes (the modes are chosen by the choice of basis $\lbrace{\ket{i}\rbrace}_{i=1}^N$ for the single-particle Hilbert space $\C^N$). In $\H_N$ we have the natural action of the fermionic creation and annihilation operators $a^\dag_i,a_i$, $i=1,\ldots,N$, satisfying the canonical anticommutation relations: $\{a_i,a_j\}=\{a_i^\dagger,a_j^\dagger\}=0$ and $\{a_i,a_j^\dagger\}=\delta_{ij}\I$, with $\I$ the identity operator on $\H_N$. 

It is convenient to introduce Majorana fermion operators $\gamma_1,\gamma_2,\ldots,\gamma_{2N}$, which are fermionic analogues of quadratures known from quantum optics and defined as $\gamma_{2i-1} =a_i+a^\dagger_i$ and $\gamma_{2i}=\mathrm{i}(a^\dagger_i -a_i)$.
Majorana operators are Hermitian and satisfy the anticommutation relations $\lbrace{\gamma_i,\gamma_j\rbrace}=2\delta_{ij}\I$ . In this work we are interested in estimating expectation values of (Hermitian) products of Majorana operators, with particular emphasis of monomials of degree two and four (pairs and quadruples). For even-sized subset $A\subset [2N]$ (using the convention $[k]=\lbrace{1,\ldots,k\rbrace}$) we define: $\gamma_A \coloneqq \mathrm{i}^{|A|/2} \prod_{i\in A} \gamma_i$, where $|A|$ denotes the size of subset $A$. We assume the convention that terms in the product are arranged in a non-decreasing manner, for example $\gamma_{1,3}=i \gamma_1 \gamma_3$ and  $\gamma_{1,3,6,11}=- \gamma_1 \gamma_3 \gamma_6 \gamma_{11}$. For any $A,B \subset [2N]$, Majorana monomials satisfy the commutation relation $\gamma_A \gamma_B= (-1)^{|A|\cdot|B|+|A\cap B|} \gamma_B \gamma_A$, which makes it impossible to measure all observables simultaneously \cite{BravyiTerhav2010}. 

 In quantum theory, a general measurement is described by a postitive operator-valued measure (POVM) i.e. a collection $\M=\lbrace{M_a\rbrace}$ of operators (called effects) satisfying $M_a\geq 0$ and $\sum_a M_a = \I$. In a given experimental run, the probability of obtaining an outcome $a$ from a measurement $\M$ on the state $\rho$ is given by the Born rule, i.e. $p(a|\M,\rho)=\tr(M_a \rho)$. An essential component of our measurement strategy is the POVM which describes the joint measurement of $N$ disjoint (and commuting) Majorana pair operators, namely the POVM with effects:
\begin{equation}\label{eq:joint_projective_measurement}
    M_{q_1,q_2,\ldots,q_N} =\prod_{i=1}^N \frac{1}{2}(\I +q_i\cdot \gamma_{2i-1,2i})\ . 
\end{equation}
In the above expression, the values $q_i\in\lbrace{\pm1 \rbrace}$, $i=1,\ldots,N$, are the outcomes of the jointly measured (commuting) quadratic Majorana observables $\gamma_{2i-1,2i}$. Alternatively we can view $\{M_{q_1,\ldots,q_N}\}$, up to a simple transformation of outcomes ($q_i \mapsto n_i=(1+q_i)/2$), as a measurement of the occupation numbers $n_i$ in $N$ fermionic modes. Another important element of our strategy  will be the class of so-called fermionic Gaussian unitaries, also known under the name of fermionic linear optics (FLO) \cite{terhal02,bravyi02,jozsa08}. Unitaries $U$ that belong to this class, which we will denote by $\mathrm{FLO}(N)$, satisfy:
\begin{equation}\label{eq:FLOrotation}
    U \gamma_i  U^\dagger = \tilde{\gamma}_i= \sum_{j=1}^{2N} R_{ji} \gamma_j\ , 
\end{equation}
where $i=1,\ldots, 2N$ and  $R_{ji}$ are the matrix elements of $R\in O(2N)$ (with $O(2N)$ denoting the group of $2N\times 2N$ orthogonal matrices). Conversely, for every $R\in O(2N)$ it is possible to find $U_R\in \mathrm{FLO}(N)$ for which Eq. \eqref{eq:FLOrotation} holds, defined up to a physically irrelevant global phase. From Eq. \eqref{eq:FLOrotation} it follows that $U_{R_1}U_{R_2}=U_{R_1 R_2}$  (up to a global phase). The $\mathrm{FLO}(N)$ transformations modify higher degree Majorana monomials in the following way:
\begin{equation}\label{eq:monomialROT}
     U \gamma_A  U^\dagger = \sum_{B\subset [2N] ,\  |B|=|A|} \det(R_{B,A})\  \gamma_B\ , 
\end{equation}
where $R_{B,A}$ is an $|A|\times |A|$ submatrix  of $R$ obtained from the subsets of rows $B$ and columns $A$ (arranged in increasing fashion). 

It is often convenient to represent a fermionic system as a system of qubits via a fermion-to-qubit mapping. A transformation of this type maps the $2N$ Majorana operators to multi-qubit Pauli operators while preserving their anticommutation relations. One such mapping is the Jordan-Wigner (JW) transformation \cite{jordan93}, in which the Majorana operators of mode $j\in [N]$ are represented as the $N$-qubit Pauli operators, $\gamma_{2j-1}=Z^{\otimes j-1}\otimes X\otimes\I_2^{\otimes N-j}$ and $\gamma_{2j}=Z^{\otimes j-1}\otimes Y\otimes\I_2^{\otimes N-j}$, where $X,Y$ and $Z$ are the qubit Pauli operators, and $\I_2$ the $2\times 2$ identity. While the Jordan-Wigner transformation can be defined irrespective of the physical arrangement of qubits, different layouts (and labellings of qubits) can have various implementation benefits (see Sec. \ref{sec:realization} for further discussions and exemplary realizations of the JW transformation in 1D and 2D layouts).

Utilizing these ingredients we will present a measurement scheme described by a single (parent) POVM that jointly measures noisy variants of all non-commuting Majorana observables.  In particular, we describe a parent POVM together with classical post-processing such that, for any even $A\subset [2N]$, we can \emph{simultaneously} reproduce the outcome statistics of $\M^{A,\eta_A}$, where
\begin{equation}\label{eq:noisymeasurement}
M_{e_A}^{A,\eta_A}=\frac{1}{2}(\I+e_A\cdot \eta_A \gamma_A)\,,
\end{equation}
with $\eta_A\in (0,1]$ the visibility and $e_A\in\{\pm 1\}$ the measurement outcome (see Fig. \ref{fig:general_idea} for an illustration of the scheme). \rev{In our work, unless otherwise stated, the term \emph{noisy} refers to unsharp (non-projective) measurements. The ``noise'' arises from a classical post-processing of the joint measurement rather than from an imperfect measurement device.}

Following the methodology presented in \cite{mcnulty22} we can construct, from the outcomes of the joint measurement, unbiased estimators,
\begin{equation}\label{eq:ESTIMATORmajorana}
    \hat{\gamma}_A\coloneq \frac{e_A}{\eta_A}\ ,
\end{equation}
of the expectation values $\tr (\gamma_A \rho)$ for an (unknown) state $\rho$. We can therefore estimate many non-commuting operators simultaneously, just as in the case of classical shadows \cite{Zhao21,Wan22,huang20}. The efficiency of the strategy for a set of observables $\mathcal{S}_k=\{\gamma_A:|A|=k\}$ is quantified by the sample complexity, i.e. the number of copies of the state $\rho$ needed to ensure, with probability at least $1-\delta$, that $|\tr(\gamma_A\rho)-\hat\gamma_{A,S}|<\epsilon$ for all $\gamma_A\in\mathcal{S}_k$, where $\hat\gamma_{A,S} = (1/S)\sum_{i=1}^S \hat\gamma^{(i)}_{A}$ is the empirical mean of the estimator $\hat\gamma_{A}$ from $S$ independent rounds. Since the estimator $\hat\gamma_A$ is binary and takes the values $\{\pm \eta_A^{-1}\}$, we can apply Hoeffding's inequality to bound the probability for each observable \cite{Hoeffding63}. Following simple reasoning (see \cite{MCO2024}) based on the union bound (over different observables $\gamma_A$), we get 
\begin{equation}\label{eq:sampleCOMPLEXITY}
S=\frac{2}{\eta_k^2 \epsilon^2}\log(2|\mathcal{S}_k|/\delta) \ ,
\end{equation}
where $|\mathcal{S}_k|=\binom{2n}{k}$ denotes the number of estimated observables and $\eta_k$ the maximal visibility for which all observables in $\mathcal{S}_k$ can be jointly measured.

For a random variable $y$ we denote by $\mathrm{Pr}_y (\mathcal{Y})$ the probability of event $\mathcal{Y}$ occurring. For a real-valued function $f$ of a random variable $y$, we will denote by $\mathbb{E}_y[f]$ the expectation value of the random variable $f(y)$. Finally, for two positive-valued functions $f(x), g(x)$ we will write
$f = \mathcal{O}(g)$ if there exists a positive constant $C$ such that $f(x)\leq C\cdot g(x)$, for sufficiently large $x$. We will write $f=\Omega(g)$ if $g=\mathcal{O}(f)$ and finally $f=\Theta(g)$ if $f=\mathcal{O}(g)$ and  $g=\mathcal{O}(f)$. 

\begin{figure}
    \centering		\includegraphics[width=8cm]{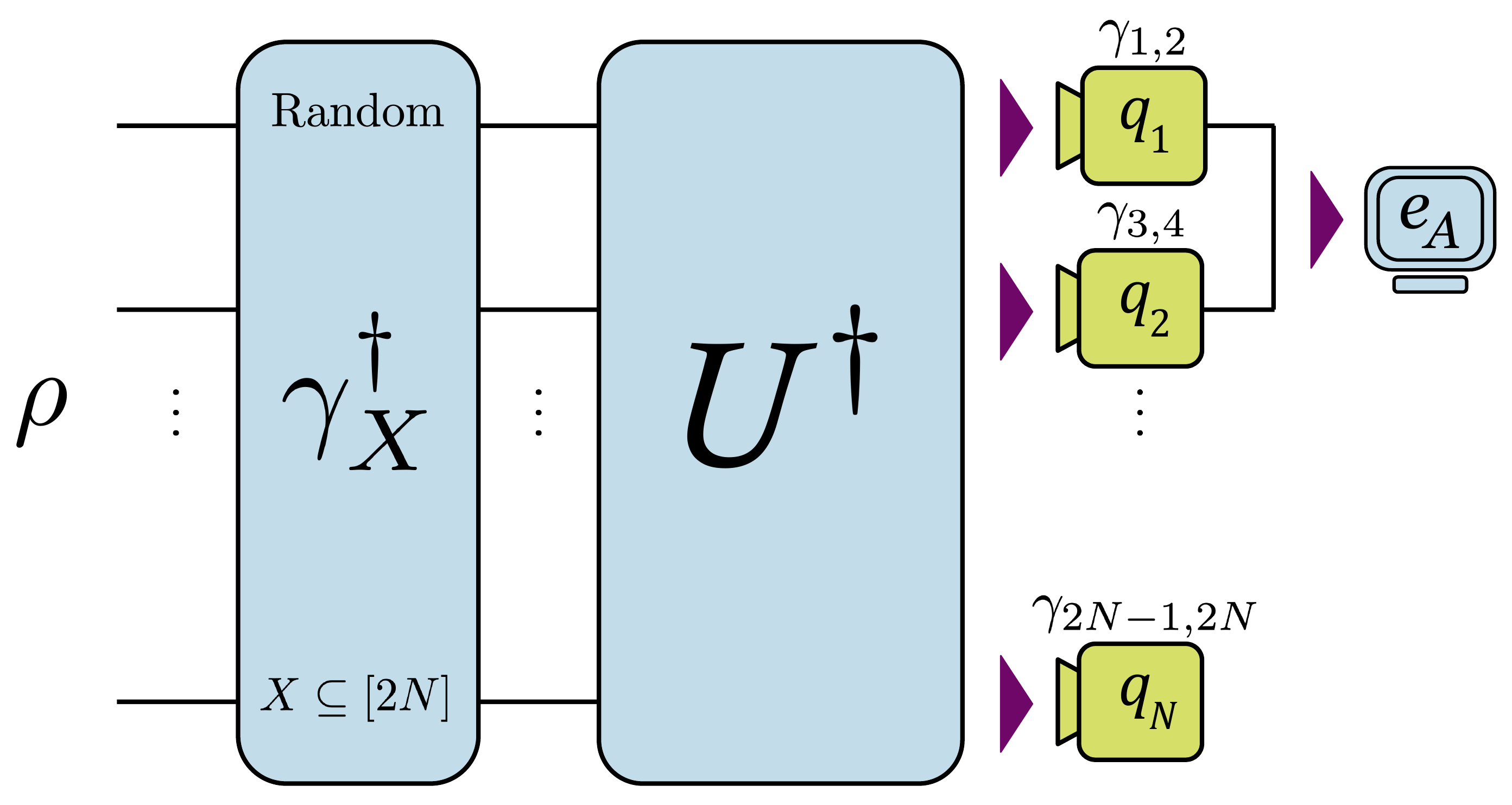}
\caption{\label{fig:general_idea}  \textbf{Circuit implementation of the joint measurement scheme for an $N$ mode fermionic state.}\\ The scheme is described in (i)--(iv) in Sec. \ref{sec:Genprotocol}. In step (i) the state $\rho$ evolves under the action of a Majorana monomial $\gamma_X^{\dagger}$, sampled uniformly at random from the set of $2^{2N}$ Majorana monomials. In step (ii) a fermionic Gaussian unitary $U^{\dagger}$ is applied. In step (iii) the $N$ projective measurements of the disjoint Majorana pairs are performed, with outcomes $q_1,\ldots,q_N$. Finally, for a given $A\subset [2N]$, one outputs $e_A$ according to step (iv).
	}
\end{figure} 

\subsection{A general joint measurement strategy} \label{sec:Genprotocol}

We now present a general strategy to simultaneously measure noisy versions of all Majorana fermion observables $\gamma_A$, with $|A|$ even, on an unknown fermionic state $\rho$. First, fix a suitably chosen $U\in\mathrm{FLO}(N)$ corresponding to $R\in O(2N)$. Let $B'\subset [N]$ label a suitably chosen subset of the $N$ commuting Majorana pairs which form the POVM $\M$ in Eq. (\ref{eq:joint_projective_measurement}), with cardinality $|B'|=|A|/2$.  Set $B=\cup_{i\in B'} \lbrace{2i-1,2i \rbrace}$ and $s_{AB}=\mathrm{sgn}(\det(R_{A,B}))$.

A single round of the measurement scheme is described as follows:

\vbox{
\begin{itemize}
    \item[(i)]  Draw $X\subseteq[2N]$ according to a uniform distribution on all $2^{2N}$ subsets of $[2N]$. Apply $\gamma^\dagger_X$ to $\rho$, obtaining $\rho'=\gamma^\dagger_X \rho \gamma_X$.
    \item[(ii)] Apply $U^\dagger$ to $\rho'$, obtaining $\rho''=U^\dagger\gamma^\dagger_X \rho \gamma_XU$.
    \item[(iii)] Simultaneously measure, on the state $\rho''$, the commuting Majorana pair operators $\gamma_{1,2},\gamma_{3,4},\ldots, \gamma_{2N-1,2N}$, obtaining outcomes $\bm{q}=(q_1,q_2,\ldots, q_N)$ with $q_i=\pm1$.
    \item[(iv)] Output $e_A=s_{AB}(-1)^{|A\cap X|} \prod_{i\in B'} q_i$.
\end{itemize}
}
\begin{prop} \label{prop:JMresult} For all $A\subset[2N]$, $|A|$--even, the above procedure generates samples $e_A$ from the POVM $\M^{A,\eta_A} $ defined in Eq. \eqref{eq:noisymeasurement}, i.e. a noisy variant of the projective measurement of $\gamma_A$, on a state $\rho$, with visibility $\eta_A=|\det(R_{A,B})|$.
\end{prop}

For the proof see Sec. \ref{sec:proof_joint_measurement}. Note that steps (i)--(iii) above are independent of the observable $\gamma_A$ whose noisy variant we aim to measure. In step (iv) we conduct an efficient post-processing that depends only on $A$ and $B'$ (which can depend on $A$) and a randomly chosen subset $X$ of $[2N]$. For this reason the POVM realized in steps (i)--(iii) is a parent POVM for noisy versions of the observables $\gamma_A$ with visibility parameter $\eta_{A}$, for all subsets $A\subset[2N]$ of even size. The outcomes can therefore be used to estimate, simultaneously, all non-commuting observables via the estimators $\hat\gamma_A$ from Eq. (\ref{eq:ESTIMATORmajorana}), since $\eta_A^{-1}\mathbb{E}_{X,\bm{q}}[e_A] = \tr(\rho \gamma_A)$. Next, to minimize the sample complexity we construct unitaries which aim to maximize the visibility of the measurements.

\subsection{Gaussian unitaries for measurements of Majorana pairs and quadruples}\label{sec:constructionOFunitaries}

\begin{figure}[h]
\centering
\includegraphics[width=8.5cm]{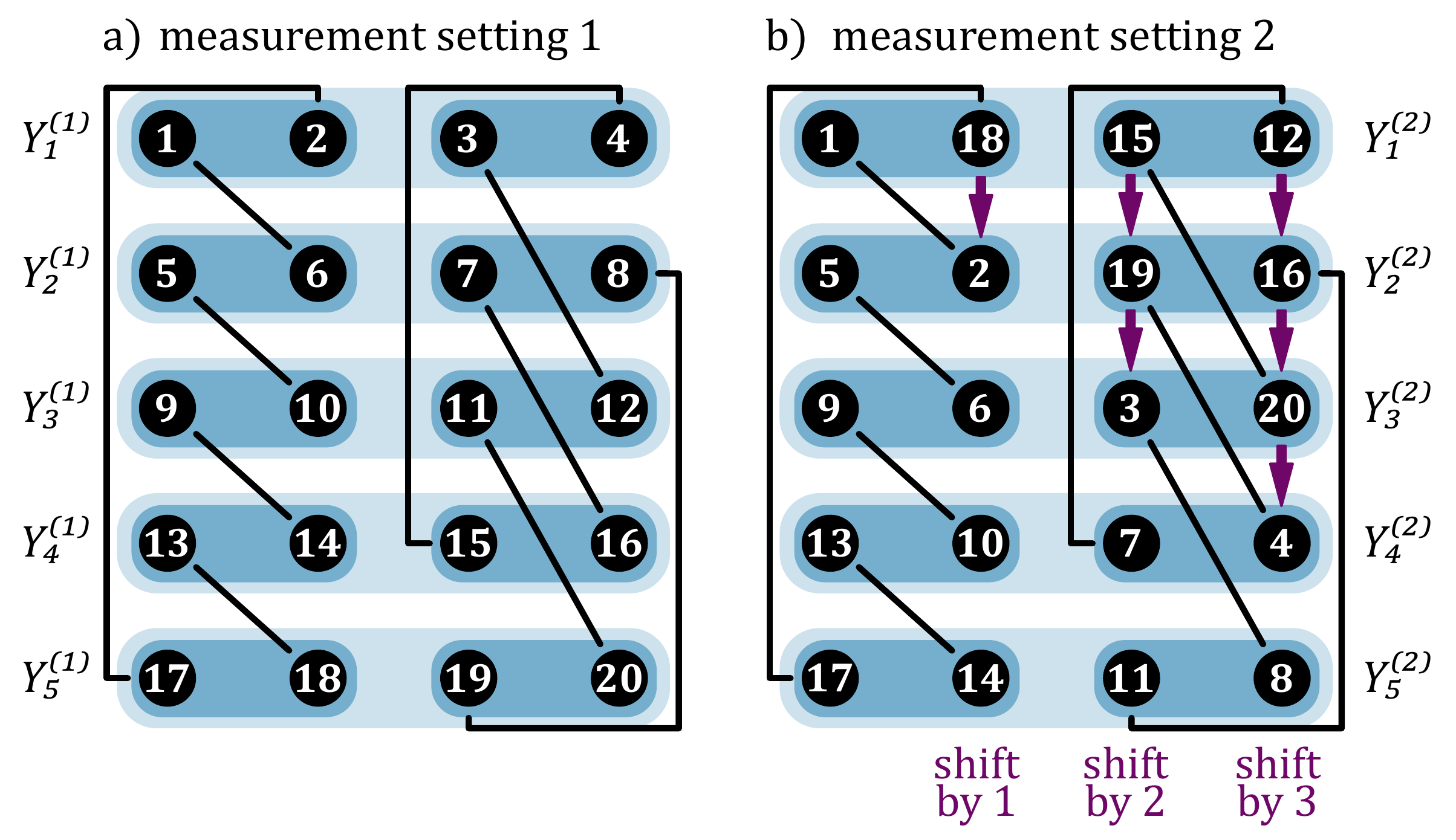}
	\caption{\label{fig:LL+1}
\textbf{Example of two measurement settings for measuring all Majorana pairs for $N=10$.} \\ The $2N=20$ Majorana operators are partitioned into 5 subsets (rows) $Y^{(r)}_{\alpha}$ (represented by light blue rounded rectangles) of cardinality 4, for each measurement setting $r=1,2$. In the first measurement setting (a), the unitary $U_{\text{sup}}$ acts via a lower flat orthogonal matrix on each row $Y_{\alpha}^{(1)}$, transforming every $\gamma_k$, $k\in Y_\alpha^{(1)}$, into a balanced superposition $\tilde{\gamma}_k=\sum_{l\in Y_\alpha^{(1)}} b^{(\alpha)}_{lk} \gamma_l$, with $b^{(\alpha)}_{lk}=\Omega(N^{-1/4})$. The pairings $(\pi(2j-1),\pi(2j))$ (marked by solid black lines) are constructed by permuting the standard pairings $\{(2j-1,2j)\}_{j=1}^N$ (represented by dark blue rounded rectangles) with a permutation $\pi$ that vertically shifts (cyclically) the second and fourth columns by 1 and 2 respectively. These pairings ensure the sample from a noisy variant of $\gamma_{ij}$ can be obtained if $i$ and $j$ are in different rows. In the second setting (b), the scheme is modified by reshuffling the operators by a permutation $\sigma^{(2)}$, that shifts columns 2,3,4 (of setting 1) vertically downwards by 1,2,3 rows, respectively. All operators in the same row of setting 1 now appear in different rows of setting 2, ensuring joint measurability of all pairs. The standard pairings represented by dark blue rectangles are consistent with the JW encoding in a 2D grid (see Fig. \ref{fig:physLayouts}), where each such rectangle corresponds to one qubit.}
\end{figure} 

To target Majorana pairs and quadruples, we propose a randomized version  of the original scheme by introducing a preliminary step (o)  in which every experimental run is preceded by sampling uniformly at random from a set of $K$ unitaries $U^{(r)}\in\mathrm{FLO}(N)$, $r=1,\ldots,K$. Then, we proceed with the usual steps (i)--(iii), as described in Sec. \ref{sec:Genprotocol}, with $U^{(r)}$ implemented in (iii). Here we sketch the construction of the $K$ unitaries that ensures that such a modification of the joint measurement from Prop. \ref{prop:JMresult} gives visibilities $\eta_2=\Theta(1/\sqrt{N})$ and $\eta_4=\Theta(1/N)$, where $\eta_k$ denotes the maximum visibility for which \emph{all} $k$--th degree Majorana monomials can be jointly measured by the parent. The details of the construction together with an extension to arbitrary constant-sized Majorana monomials are presented in \cite{MCO2024}. 

The unitaries we propose consist of the product of three unitaries:
\begin{equation}
\label{eq:U_sup_pair_decomp}
  U^{(r)} =U^{(r)}_{\mathrm{resh}}   U_{\mathrm{sup}} U_{\mathrm{pair}} \ ,
\end{equation}
where $U_{\mathrm{pair}}$ and $U^{(r)}_{\mathrm{resh}}$ correspond to suitably chosen permutations $R_\pi,R_{\sigma^{(r)}}\in O(2N)$, respectively, and $U_{\mathrm{sup}}$ realizes a \emph{balanced superposition} (see Fig. \ref{fig:LL+1}) of Majorana modes belonging to disjoint subsets $Y_\alpha$ which form a partition of the set of all modes such that $\cup_{\alpha=1}^L Y_\alpha =[2N]$. 
In the rest of this section we focus on a specific case in which $2N=L(L+1)$ with the Majorana modes divided into $L+1$ subsets $\{Y_\alpha\}_{\alpha=1}^{L+1}$, each with cardinality $|Y_\alpha|=L$. If this relation is not satisfied, it is possible to either add $2N_{aux}\leq 2L$ auxiliary Majorana modes to embed a smaller fermionic system into a larger one, or to drop the restriction that every $Y_\alpha$ has the same size (see \cite{MCO2024} for details).

The role of $U_{\mathrm{pair}}$ is to effectively realize a different pairing of the Majorana modes in measurement \eqref{eq:joint_projective_measurement} such that

\begin{equation}
    \tilde{M}_{q_1\ldots q_N} = U_{\mathrm{pair}} M_{q_1\ldots q_N} U_{\mathrm{pair}}^\dagger    =\prod_{i=1}^N \frac{1}{2}(\I +q_i \gamma_{\pi(2i-1),\pi(2i)})\ .
\end{equation}
An essential part of the construction is to ensure each pairing $(\pi(2i-1),\pi(2i))$ connects two distinct subsets of Majorana modes, i.e. $\pi(2i-1)\in Y_{\alpha}$ and $\pi(2i)\in Y_{\beta}$, such that every pair of subsets is connected once. See Fig. \ref{fig:LL+1} for an example of a pairing and divisions into disjoint subsets for $2N=20$.

In order to describe the action of $U_{\mathrm{sup}}$ we will abuse notation and denote by $Y_\alpha$ subspaces spanned by vectors $\{\ket{i}\}_{i\in Y_\alpha}$.  The orthogonal transformation behind $U_{\mathrm{sup}}$ is of the form
\begin{equation}\label{eq:SuperposingOrthogonal}
R_{\mathrm{sup}}=\oplus_{\alpha=1}^{L+1} R^{(\alpha)} \ ,
\end{equation}
where  $R^{(\alpha)}$ are operators acting on spaces $Y_\alpha$ constructed from $(L$$)\times (L)$ \emph{lower-flat} orthogonal matrices  defined in \cite{MCO2024}, i.e., matrices with entries satisfying $R^{(\alpha)}_{ij}=\Omega(1/\sqrt{L})=\Omega(1/N^{1/4})$ for $i,j\in Y_\alpha$. Note that orthogonal matrices with this property exist in every dimension \cite{Jaming15}.

Finally, for $r>1$, the permutation $\sigma^{(r)}$ (which is not unique) corresponding to $U^{(r)}_{\mathrm{resh}}$ acts by reshuffling the modes into different partitions $\sigma^{(r)}(Y_\alpha):=Y_\alpha^{(r)}$ of $[2N]$, for every $\alpha$. For $r=1$ we set $U^{(r)}_{\mathrm{resh}}=\I$ and label $Y_\alpha^{(1)}:=Y_\alpha$.

Performing the measurement with $U^{(r)}$ defined above, we can verify (by directly repeating steps in the proof of Prop. \ref{prop:JMresult}) that every Majorana pair $\gamma_{ij}$, with indices $i,j$ belonging to different subsets, $i\in Y_{\alpha_i}^{(r)}, j\in Y_{\alpha_j}^{(r)}$,  can be jointly measured with visibility $\eta_{ij}=|R^{(\alpha_i)}_{ki}R^{(\alpha_j)}_{lj}|$, where $(k,l)$ is the pairing between subsets $Y_{\alpha_i}^{(r)}$ and $Y_{\alpha_j}^{(r)}$ (note that by construction $l\in Y_{\alpha_i}^{(r)}$ and $k\in Y_{\alpha_j}^{(r)}$). Using the flatness property of the rotations $R^{(\alpha)}$, we get $\eta_{ij}=\Theta(1/\sqrt{N})$. By analogous reasoning, a four-element subset $A=\{ijkl\}$ together with the same unitary and pairings described above, results in $\eta_A=\Theta(1/N)$, as long as the individual elements $i,j,k,l$ belong to different subsets $Y_\alpha^{(r)}$. 

In general not all pairs and quadruples will consist of elements belonging to different subsets $Y_\alpha^{(r)}$. We denote by $\mathcal{M}^{(r)}$ the set of Majorana pairs and quadruples that are \emph{consistent} with the partition structure inherent to the unitary $U^{(r)}$, i.e. $A\in\mathcal{M}^{(r)}$ if $A$ contains elements only from distinct partition subsets. After implementing steps (o)--(iii) of the measurement, if $A\in\mathcal{M}^{(r)}$ we continue to step (iv), otherwise we perform (iv') by sampling the outcome $e_A\in\{\pm 1\}$ at random with equal probability.

This modification realizes a joint measurement of all pairs and quadruples $A$ with visibilities $\eta_{ij} \geq (1/K_{2}) \tilde{\eta}_{ij}$ and $\eta_{ijkl}\geq (1/K_4) \tilde{\eta}_{ijkl}$, where $\tilde{\eta}_A$ is the visibility attainable from the unitary whose partition is consistent with $A$, and $K_2,K_4$ are the minimal numbers of measurement settings that ensure all pairs and quadruples are compatible with at least one measurement round, respectively.  In Fig. \ref{fig:LL+1} we give an example of two measurement settings ($K_2=2$) that ensure all  pairs are measured for a system of $2N=L(L+1)$ Majorana modes arranged in a rectangular layout (with $L$--even). It can also be shown that when $L+1$ is a prime number, all Majorana quadruples can be measured using $K^{\text{prime}}_4=7$ settings via a method that is compatible with the original layout and can be obtained by permuting Majorana modes in columns of the Majorana fermion arrangement from Fig. \ref{fig:LL+1} (part a). See Supplementary Information \ref{app:quadropNumbr} for the proof. For general $L$, a probabilistic reasoning given in \cite{MCO2024} ensures that $K_4=9$ random partitions of the $2N=L(L+1)$ modes into subsets of size $L$ (not necessarily compatible with the 2D layout) suffice to measure all quadruples. Furthermore, we can extend this to arbitrary $N$ mode systems by relaxing the constraint that the subsets are equally sized.

\begin{prop}\label{prop:Randomized}
     The randomized version of the measurement protocol defined in Sec. \ref{sec:Genprotocol} with at most 9 different FLO unitaries $U^{(1)},\ldots,U^{(9)}$ jointly measures the noisy variants  $\M^{A,\eta_A}$ for all Majorana pairs and quadruples with 
     
\begin{equation}\label{eq:visibility_quadratic_quart}
\eta_A=\begin{cases}
{\Theta}(N^{-1/2}) &\mbox{if} \,\,\, |A|=2 \\
{\Theta}(N^{-1}) &\mbox{if} \,\,\, |A|=4 \,.
 \end{cases}
\end{equation}
 
\end{prop}

The above result, taken together with Eq. \eqref{eq:sampleCOMPLEXITY} and the observation that we have $|\mathcal{S}_2|=\binom{2N}{2}={\Theta}(N^2)$ Majorana pairs and $|\mathcal{S}_4|=\binom{2N}{4}={\Theta}(N^4)$ Majorana quadruples, gives the claimed sample  complexities $S_2$ and $S_4$ for estimating quadratic and quartic Majorana observables respectively:
\begin{equation}\label{ep:sampleCompl}
    S_2=   \mathcal{O}(N \log(N)/\epsilon^2)\ ,\ S_4= \mathcal{O}(N^2 \log(N)/\epsilon^2) \ .
\end{equation}

It is shown in \cite{MCO2024} that the joint measurement of Prop. \ref{prop:Randomized} is asymptotically optimal. In other words, the optimal visibility for which the set of observables $\mathcal{S}_k$ can be jointly measured (via \emph{any} parent) satisfies $\eta_k=\Theta(N^{-k/4})$ for $k=2,4$. Therefore, it is not possible to construct a joint measurement which reduces the sample complexity by more than a constant factor.

   In general, in a 2D rectangular $L_x \times L_y$ layout of Majorana modes (see e.g. Fig. \ref{fig:LL+1}, also Fig. \ref{fig:rounds} in Sec. \ref{sec:physical_hamiltonians}) we can refer to transformations which act vertically or horizontally. A \emph{vertical (horizontal)} FLO transformation is one whose corresponding orthogonal matrix $O$ can be written as:
   \begin{equation}\label{eq:vertical_horizontalO}
    O= \bigoplus_{i=1}^{L_x (L_y)} O_i, 
    \end{equation}
    where $O_i$ acts on modes within the $i$-th column (row) in the 2D layout. An important property of $U^{(r)}$ from \eqref{eq:U_sup_pair_decomp} is that it decomposes into a product of a horizontal transformation $U_{\mathrm{sup}}$ and vertical transformations $U_{\mathrm{pair}}$, $U_{\mathrm{resh}}^{(r)}$.  This makes it feasible for experimental implementations in two-dimensional qubit architectures under the JW encoding in gate depth $\mathcal{O}(N^{1/2})$ and with two-qubit gate count $\mathcal{O}(N^{3/2})$, as explained in Sec. \ref{sec:realization}.

\begin{figure*}[]
\centering
\includegraphics[width=18cm]{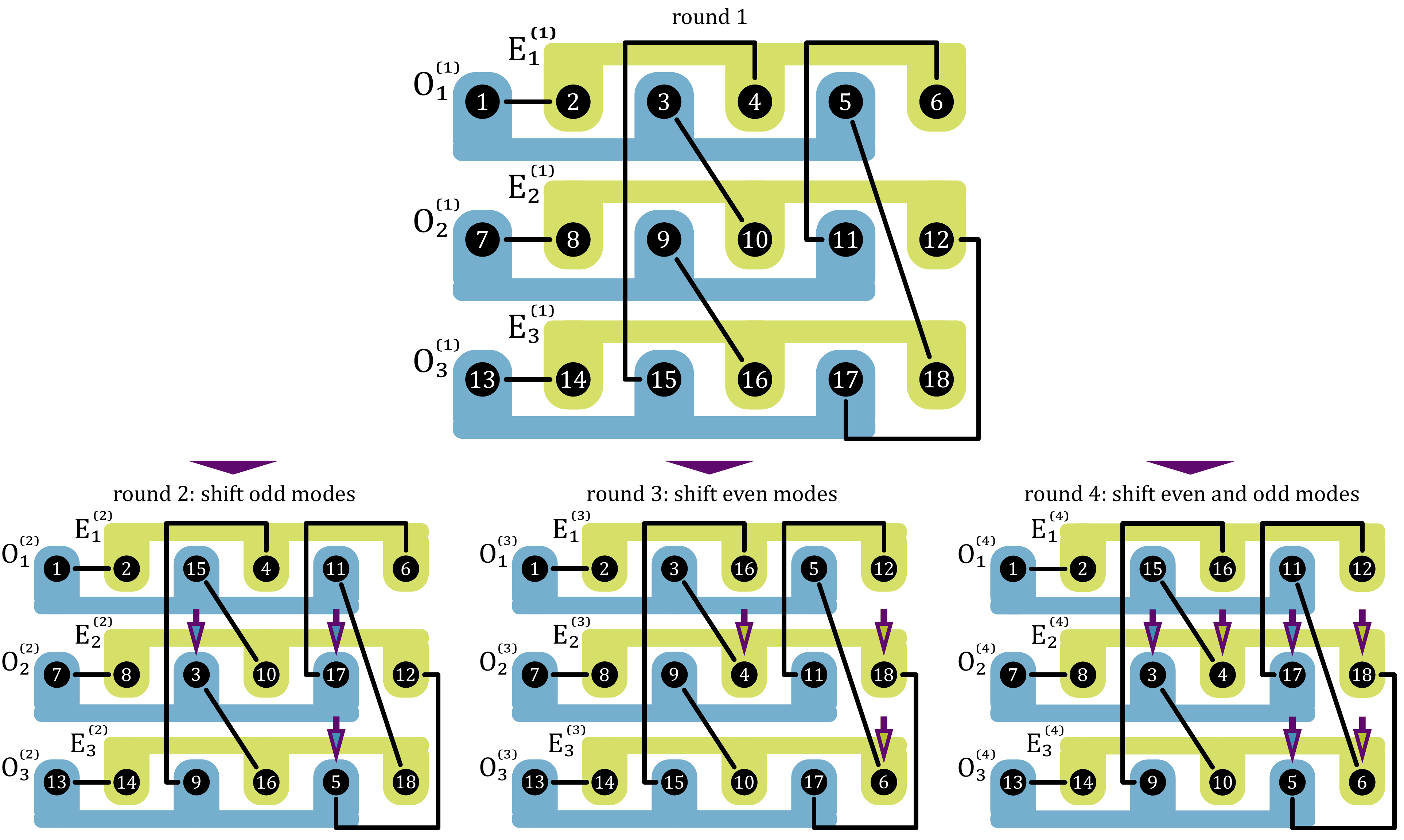}
\caption {\textbf{Schematic representation of the construction of the 4 rounds of our measurement strategy for physical fermionic Hamiltonians.}\\
The Hamiltonians are described in Sec. \ref{sec:physical_hamiltonians}. The $2N$ Majorana modes are arranged on an $2L \times L$ grid and partitioned into $L$ even ($E_{\alpha}$) and $L$ odd ($O_{\alpha}$) subsets (clusters) of cardinality $L$. The unitary $U_{\text{sup}}$ acts via a lower-flat orthogonal matrix on each cluster. In round 1, the pairings $(\pi(2j-1),\pi(2j))$ which define $U_{\text{pair}}$ (and marked by solid black lines) are constructed by a permutation $\pi$ that shifts the second and third even column cyclically by 1 and 2 respectively (in a manner analogous to Fig. \ref{fig:LL+1}). Importantly, the pairings (black lines) ensure each even cluster is coupled to each odd cluster. In rounds $r=2,3,4$, the different partitions of the modes into even and odd clusters are implemented by $U^{(r)}_{\text{resh}}$ via the permutations $\sigma^{(r)}$. The permutation $\sigma^{(2)}$ is realized by a vertical cyclical shift of the second and third odd columns of the first measurement setting by 1 and 2, respectively (bottom left). Similarly, $\sigma^{(3)}$ corresponds to a vertical cyclical shift of the second and third even columns of the first measurement (bottom middle). Finally, $\sigma^{(4)}$ is realized by performing  both the vertical shifts realizing $\sigma^{(2)}$ and $\sigma^{(3)}$  (bottom right).}
    \label{fig:rounds}
\end{figure*}

\subsection{Strategy for physical fermionic Hamiltonians}
\label{sec:physical_hamiltonians}

We now present a measurement strategy tailored to estimate expectation values of Hamiltonians describing electronic degrees of freedom in quantum chemistry systems. The Hamiltonians can be expressed as a linear combination of a subset of quadratic and quartic Majorana fermion observables \cite{Zhao20}:
\begin{equation}\label{eq:chemistryHamiltonian}
    H=\sum_{A\in \mathcal{X}_{2}\cup \mathcal{X}_{4}} h_A \gamma_A\ , 
\end{equation}
where $h_A\in\mathbb{R}$. The set $\mathcal{X}_{2}\subset \binom{[2N]}{2}$ is the set of $2$-element subsets of $[2N]$ containing an even and odd integer, and $\mathcal{X}_{4}\subset \binom{[2N]}{4}$ is the set of 4-element subsets of $[2N]$ containing 2 even and 2 odd integers. In other words, in quantum chemistry Hamiltonians, even Majorana operators couple only to odd Majorana operators whereas the only possible couplings for quadruples are those with two even and two odd operators \cite{Zhao20}. This follows from the observation that electrons have only two-body interactions (via the Coulomb potential) and there is no spin-dependant interaction or potential affecting them. It is henceforth sufficient to design a strategy which estimates noisy versions of Majorana pair and quadruple operators restricted to those in $\mathcal{X}_2$ and $\mathcal{X}_4$.

To construct an estimator for $\tr(H \rho)$ we propose four measurement \textit{rounds} (settings). In contrast to the randomized scheme in  Sec. \ref{sec:constructionOFunitaries}, the four FLO unitaries $U^{(r)} =U^{(r)}_{\mathrm{resh}}U_{\mathrm{sup}} U_{\mathrm{pair}}$ are implemented deterministically. We consider the system size $2N = 2L^2$ which allows for the arrangement of a rectangular $L\times L$ lattice of qubits encoding the $L^2$ fermionic modes (and $2L^2$ Majorana modes) via the Jordan-Wigner transformation (see Fig. \ref{fig:rounds} for an example for $L=3$). Systems which don't satisfy this can always be encoded in larger ones.

The distinct aspect of this scheme is the partitioning of the $2L \times L$ grid  representing Majorana operators into even and odd subsets (clusters) labelled $E^{(r)}_\alpha$ and $O^{(r)}_\alpha$ (and coloured green and blue in Fig. \ref{fig:rounds}), each containing $L$ Majorana modes. The FLO unitaries which implement the measurement are described as follows:

\begin{itemize}
    \item[(a)]  The permutation $\pi$ defining $U_{\text{pair}}$ translates the standard pairings to $\{(\pi(2j-1),\pi(2j))\}_{j=1}^N$ (marked by black lines in Fig. \ref{fig:rounds}) via vertical shifts of the columns of the $2L \times L$ grid. In particular, $\pi$ cyclically shifts the $\ell$--th even column vertically downwards by $\ell-1$ rows. The resulting pairs ensure all even clusters are connected to all odd clusters exactly once.
    \item[(b)] The unitary $U_{\text{sup}}$, in analogy with Eq. (\ref{eq:SuperposingOrthogonal}), acts via $R_{\text{sup}}=\oplus_{\alpha=1}^{2L}R^{(\alpha)}$, where $R^{(\alpha)}\in O(L)$ is a lower-flat orthogonal matrix acting on an even or odd cluster.
     \item[(c)] The permutation $\sigma^{(r)}$ defining $U_{\text{resh}}^{(r)}$, maps the original clusters $E^{(1)}_\alpha$ and $O_\alpha^{(1)}$ into $\sigma^{(r)}(E^{(1)}_\alpha):=E_{\alpha}^{(r)}$ and $\sigma^{(r)}(O_\alpha^{(1)}):=O_{\alpha}^{(r)}$ using vertical shifts within columns of the 2D layout. The permutation $\sigma^{(2)}$ is realized by cyclically shifting the $\ell$--th even column vertically downwards by $\ell-1$ rows, while $\sigma^{(3)}$ corresponds to the same shifts acting on the odd columns. Finally, $\sigma^{(4)}=\sigma^{(2)}\circ\sigma^{(3)}$ implements both $\sigma^{(2)}$ and $\sigma^{(3)}$ simultaneously (as illustrated by the blue and green arrows in Fig. \ref{fig:rounds}).
\end{itemize}

In Supplementary Information \ref{app:4rounds} we show that these four unitaries $U^{(r)}$ are sufficient to estimate all quadratic and quartic terms in $H$ with visibilities from Prop. \ref{prop:Randomized}. In particular, for $A\in\mathcal{X}_2\cup\mathcal{X}_4$, there exists at least one measurement round $r$ in which $A\in\mathcal{M}^{(r)}$ and therefore $\eta^{(r)}_A=\Theta(N^{-|A|/4})$. Note that because the FLO unitaries appearing in points (a)-(c) are either horizontal or vertical (see \eqref{eq:vertical_horizontalO} for the definition), they can be implemented with the same gate depth and gate count as those discussed in Sec. \ref{sec:constructionOFunitaries}.

A ``single--shot'' estimator for the Hamiltonian, constructed from the outcomes of the four measurement rounds, is defined as
\begin{equation}
\label{eq:hamiltonian_estimator}
    \hat{H} = \sum _{r=1}^4\sum_{A\in \mathcal{X}_2\cup\mathcal{X}_4}  h_A \alpha_{A}^{(r)}\hat{\gamma}_A^{(r)}\ ,
\end{equation}
where
\begin{equation}\label{eq:estimator_av}
\hat{\gamma}_A^{(r)}=\begin{cases}
e^{(r)}_{A}/\eta^{(r)}_{A} &\mbox{if} \,\,\, A \in \mathcal{M}^{(r)} \\
0 &\mbox{otherwise} \,,
 \end{cases}
\end{equation}
and $e^{(r)}_{A}\in\{\pm 1\}$ is the outcome obtained from the post-processing in step (iv) of the general protocol with unitary $U^{(r)}$ (cf. Sec. \ref{sec:Genprotocol}). We choose coefficients $\alpha^{(r)} \in \mathbb{R}$ satisfying $\sum_{r:A\in \mathcal{M}^{(r)}} \alpha_A^{(r)}=1$, which ensures $\ev{\sum_{r} \alpha_A^{(r)} \hat{\gamma}_A^{(r)}} = \mathrm{tr}(\gamma_A\rho)$ and therefore $\hat H$ is an unbiased estimator of $\mathrm{tr}(H\,\rho)$.
For details of the post-processing see Supplementary Information \ref{sec:details_strategy_physical_ham}. For details of the coefficient choice see Supplementary Information \ref{app:optimization}.
To improve the efficiency of our estimation strategy we can apply the median-of-means method (in analogy with classical shadows \cite{huang20}). This provides a simple classical post-processing in order to reduce the effect of estimation errors, and the sample complexity is bounded by a multiple of the variance $\mathrm{Var}[\hat H]$ of the estimator, which for our scheme is given explicitly in Supplementary Information \ref{app:variance}.

\subsection{Numerical benchmarks}\label{sec:numerical_benchmarks}

\begin{table}[b]
\begin{tabular}{|c|c|c|c|c|c|c|}
\hline
\begin{tabular}[c]{@{}c@{}}Molecule\\ (qubits)\end{tabular} & \begin{tabular}[c]{@{}c@{}} CS-Pauli \\ \cite{huang20} \end{tabular}& \begin{tabular}[c]{@{}c@{}}LBCS\\ \cite{hadfield20}\end{tabular}   & \begin{tabular}[c]{@{}c@{}}BRG\\ \cite{Huggins21}\end{tabular}  & \begin{tabular}[c]{@{}c@{}}CS-FGU \\ \cite{Zhao21} \end{tabular} & \begin{tabular}[c]{@{}c@{}}This work\\  \end{tabular} \\ \hline
H2 (8)                                                      & 51.4     & 17.5 & 22.6 & 69.6 & 26.6                                            \\ \hline
LiH (12)                                                    & 266      & 14.8 & 7.0    & 155      & 106                                              \\ \hline
BeH2 (14)                                                   & 1670     & 67.6 & 68.3 & 586      &  245                                              \\ \hline
H2O (14)                                                    & 2840     & 257  & 6559 & 8440     & 1492                                               \\ \hline
NH3 (16)                                                    & 14400    & 353  & 3288 & 5846     & 1061                                                 \\ \hline
\end{tabular}
\caption{\textbf{Comparison of variances \rev{of a single measurement round} for estimating energies of standard benchmark molecules.}\\ The variances are calculated for the ground states of five molecular Hamiltonians, with units Ha$^2$. The first four strategies (taken from a summary in \cite{Zhao21}) include: classical shadows with Pauli measurements (CS-Pauli \cite{huang20}); locally biased classical shadows (LBCS \cite{hadfield20}); basis-rotation grouping (BRG \cite{Huggins21}); and classical shadows with fermionic Gaussian unitaries (CS-FGU \cite{Zhao21}). Where applicable (\cite{huang20, hadfield20, Huggins21}), the Jordan-Wigner encoding was used. The last column summarizes our results, optimized over the coefficients $\alpha_A^{(r)}$. \rev{To ensure a fair comparison of the sample complexities, the variances in the last column should be multiplied by a factor of four to account for the need of four independent measurement rounds to estimate the energy.}}
\label{tab:benchmark}
\end{table}

We benchmark our strategy by calculating the variance of our estimator \rev{(using formula \eqref{eq:molecular_variance})} on the ground state of molecular Hamiltonians that are commonly used in the literature (see e.g. \cite{Zhao21, Wan22, hadfield20, Huggins21, ogorman22, mcnulty22, gresch23}) \rev{and allow for a comparison with a wide range of strategies}. Table \ref{tab:benchmark} presents a comparison between the variances resulting from our strategy and other measurement approaches, including fermionic classical shadows \cite{Zhao21}. \rev{It is standard in benchmarks \cite{hadfield20, Zhao21} to examine the variance as an indicator of the effectiveness of a strategy.}

For the calculations in Table \ref{tab:benchmark}, the coefficients $\alpha_A^{(r)}$ of the estimator \eqref{eq:hamiltonian_estimator} are chosen as optimal (i.e., minimizing the variance) for the ground state of a given Hamiltonian. \rev{The optimized coefficients are only applicable if prior information about the state and Hamiltonian is known.} However, state-independent coefficients---uniform over rounds---can also be chosen, \rev{and they perform comparatively for larger molecules}. For details of the optimization and coefficient choice, see Supplementary Information \ref{app:optimization}. Note that our estimator requires four distinct measurement rounds. \rev{Using standard concentration arguments, the sample complexity is
bounded by a multiple of the variance \cite{huang20}, by e.g. Hoeffding's inequality}. Thus a fair comparison \rev{of sample complexities against} strategies using single round estimators requires multiplying our variance by four.
For further technical details \rev{on the numerical benchmarks} see Supplementary Information \ref{app:numerics}. 

\rev{To analyze the results of Table \ref{tab:benchmark}, we first note that BRG and LBCS are strategies tailored to specific Hamiltonians. It can be expected that for more complex Hamiltonians, and larger system sizes, their advantage would diminish, due to the large number of terms in the Hamiltonian's decomposition, making the bias towards certain observables in the decomposition less beneficial. Moreover, the optimized strategies require numerical optimizations, prior knowledge of the Hamiltonian, and a classical approximation of the underlying quantum state. For these strategies, experiments are tailored towards a specific Hamiltonian, and therefore estimating other observables with the same data is not expected to perform well. This is not the case for non-optimized strategies, including our scheme. The optimization we propose is optional and performed only in post-processing, hence it can be performed independently on each new observable of interest.}

By taking advantage of the fact that FLO circuits can be efficiently classically simulated  \cite{terhal02, Bravyi_2012}, we also implemented a simulation of our strategy on Gaussian states, whose energies can be easily calculated analytically (see e.g. \cite{Melo_2013}). \rev{For details, see Supplementary Information \ref{app:gaussian_states_simulation}}. We used the package FLOYao.jl \cite{bosse_2022} to classically simulate individual runs of our protocol, verifying (up to statistical precision) that the strategy correctly estimates the above molecular Hamiltonians on random Gaussian states. The empirical variance of this estimation also confirms (within statistical precision) the accuracy of the variance formula derived in Supplementary Information \ref{app:variance}.

\section{Discussion}
\label{sec:discussion}
We have introduced a scheme to estimate non-commuting fermionic observables based on the implementation of a joint measurement of noisy versions of products of  Majorana fermion observables. We focused on products of degree two and four, motivated by Hamiltonians from quantum chemistry. For these observables and standard benchmark Hamiltonians, our scheme exhibits sample complexity guarantees comparable with fermionic classical shadows.

The three distinctive features of our strategy, which make it attractive for practical purposes, are the  low number of complicated measurement settings required, low circuit depth in certain lattices, and the dependence of the observables' estimates on single-qubit measurements on few qubits. First, in our scheme to measure all degree two and four Majorana monomials with visibility $\Theta(N^{-{1/2}})$ and $\Theta(N^{-{1}})$, respectively, we need at most  nine measurement settings. For the special case of electronic structure Hamiltonians relevant for quantum chemistry, as defined in Sec. \ref{sec:physical_hamiltonians}, we need four measurement settings to estimate the energy. In both cases an additional randomization over Majorana products is required,  however this can be implemented easily, by a single gate layer under any fermion-to-qubit mapping. Second, if the system gets encoded in a 2D rectangular lattice via the Jordan-Wigner encoding, our strategy can be implemented in gate depth $\mathcal{O}(N^{1/2})$ with $\mathcal{O}(N^{3/2})$ 2-qubit gates. This compares favorably with fermionic and matchgate classical shadows, which both require depth $\mathcal{O}(N)$ with $\mathcal{O}(N^{2})$ 2-qubit gates. Third, our strategy uses single-qubit measurement outcomes from only one or two qubits to estimate Majorana pairs or quadruples, respectively. As such it may readily be combined with randomized error mitigation techniques \cite{berg_2022}.

\rev{We expect our method to be suitable for NISQ devices. In the 2D setting, the gate scaling advantage is expected to become apparent for system sizes of the order of hundreds of qubits due to the overhead in gate depth and count incurred by implementing the vertical FLO transformations (see discussion in Section \ref{sec:realization}). Moreover, for these system sizes, ease of error mitigation for our method persists, no classical optimization is necessary, and optimized methods are expected to yield less advantage for more complex Hamiltonians. Finally, in current superconducting architectures, implementing arbitrary randomized circuits quickly (especially those required by fermionic classical shadows) is a non-trivial challenge. In contrast, randomized single-layer Paulis can be effectively implemented due to the prevalence of randomized benchmarking and mitigation methods.}

There are several potential directions in which this work can be developed further. One interesting question to consider is how the choice of unitaries in the measurement process, discussed in Sec. \ref{sec:constructionOFunitaries}, can be optimized for a particular Hamiltonian, in order to estimate the relevant products with higher visibility. The choice of unitary can significantly affect the variance, as demonstrated in Table \ref{tab:benchmark_appendix} of Supplementary Information \ref{app:optimization}. Optimization methods have already led to substantial reductions in the sample complexity of classical shadow protocols, especially in the locally biased regime (see Table \ref{tab:benchmark}). Moreover, in this work we only focus on degree two and four Majorana operators. A natural extension would be to consider joint measurement strategies for higher degree Majorana observables, as well as to investigate whether other layouts, fermion-to-qubit encodings, or more general orthogonal transformations, provide performance benefits. Furthermore, the effect of noisy physical implementations on the scheme could be studied, similarly to \cite{mcnulty22}. Finally, our strategy can be directly applied on native fermionic computing or simulation platforms, e.g. programmable neutral atom arrays \cite{gonzalez-cuadra_fermionic_2023}, without the need to map to qubits, and is naturally compatible with devices using a 2D layout. An interesting research direction would be to experimentally apply it on such devices.

\section{Materials and methods}
\label{sec:materials_methods}
\subsection{Proof for joint measurement strategy}
\label{sec:proof_joint_measurement}
Below we provide the proof of Proposition \ref{prop:JMresult}:
\begin{proof}
We first note that
\begin{equation}\label{eq:auxTRANSF}
    \gamma_X  U \gamma_{B} U^\dagger \gamma^\dagger_X =\sum_{\substack{Y\subset [2N]\\|Y|=|B|}} \det(R_{Y,B}) (-1)^{|X\cap Y|}  \gamma_Y\ , 
\end{equation}
where we have used \eqref{eq:monomialROT} and the commutation relation of the operators $\gamma_X$ and $\gamma_Y$. In order to prove the claim in the proposition it suffices to show that for $x=\pm 1$ we have $\mathrm{Pr}_{X,\bm{q}}(e_A=x)=\tr(\rho \frac{1}{2}(\I+x\cdot \eta_{A} \gamma_{A}  ))$, where $X$ and $\bm{q}$ are random variables on which $e_A$ depends. Using the definition of $e_A$ in step (iv) and the fact that the distribution of $X$ is uniform we have 
\begin{eqnarray}
  \mathrm{Pr}_{X,\bm{q}}(e_A=x)  =  \mathrm{Pr}_{X,\bm{q}}\left( \prod_{i\in B'} q_i =x s_{AB} (-1)^{|A\cap X|}\right) \\
   = \frac{1}{2^{2N}}  \sum_{X\subseteq [2N]} \tr\left( \rho'' \frac{1}{2}(\I + x s_{AB} (-1)^{|A\cap X|} \gamma_B)  \right), \label{eq:aux2}
\end{eqnarray}
where we have also used the fact that  $q_{B'}:=\prod_{i\in B'}q_i$ obtained from step (iii) realizes the projective measurement of   $\prod_{i\in B'} \gamma_{2i-1,2i}=\gamma_{B}$. By using the definition of $\rho''$  and switching to the Heisenberg picture in which the unitary $\gamma_X U$ is applied to the POVM element in Eq. \eqref{eq:aux2} we obtain that $ \mathrm{Pr}_{X,\bm{q}}(e_A=x)$ equals
   \begin{eqnarray*}
  \frac{1}{2^{2N+1}} \sum_{X\subseteq [2N]}  \tr\left( \rho (\I + x s_{AB} (-1)^{|A\cap X|} \gamma_X  U \gamma_{B} U^\dagger \ \gamma^\dagger_X  \right)\ .
\end{eqnarray*} 
Utilizing \eqref{eq:auxTRANSF} we find that, in the above expression, the part which is proportional to the variable $x$ equals
\begin{eqnarray}
    & \frac{s_{AB}}{2^{2N+1}} \sum_{\substack{X,Y\subseteq [2N]\\ |Y|=|B|}}  \tr\left(\rho  (-1)^{|A\cap X|+|X\cap Y|} \gamma_Y \det(R_{Y,B})\right) \nonumber \\
    & = \frac{1}{2}\eta_{A} \tr(\rho \gamma_A)\ , 
\end{eqnarray}
where we have used that for $|A|,|B|$--even: 
\[
\sum_{X,Y\subseteq [2N], |Y|=|B|} (-1)^{|A\cap X|+|X\cap Y|}=2^{2N}\delta_{A,Y}\  .\]
\end{proof}

\subsection{Implementation of measurement strategies}\label{sec:realization} 

\begin{figure}[h]
    \centering
\includegraphics[width=8cm]{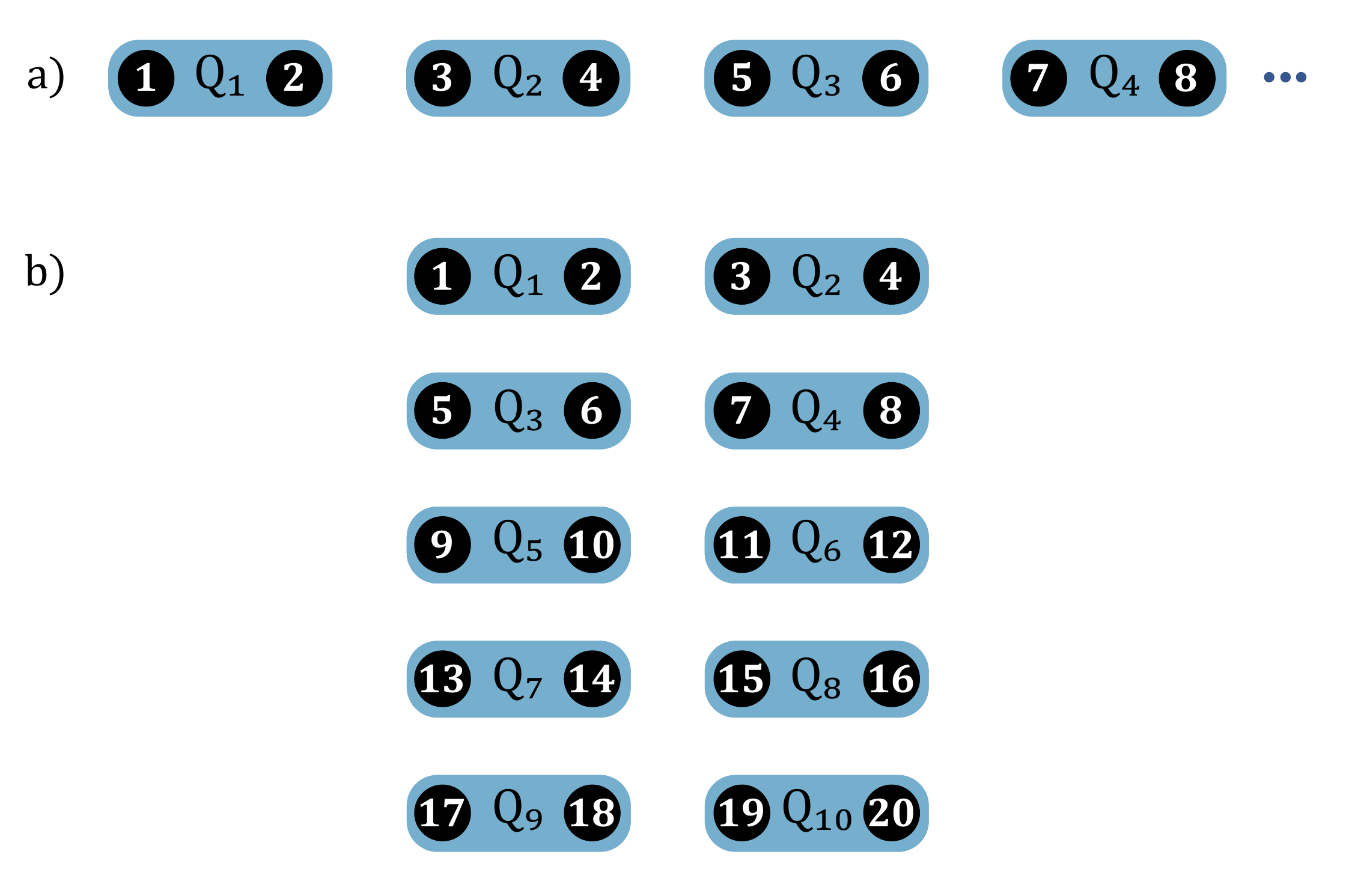}
	\caption{\label{fig:physLayouts}
\textbf{Two exemplary realizations of the Jordan-Wigner encoding in systems of qubits that exhibit different spatial locality.} \\ a) 1D array and b) $L/2 \times (L+1)$ rectangular qubit array (for $L=4$). The individual qubits are indicated by consecutive pairs within the blue rounded rectangles labelled by $Q_1,Q_2,\ldots$, each enclosing two Majorana mode operators 
 indicated by black dots.  Importantly,  locality in Majorana modes (e.g. coupling $\exp(\mathrm{i} \varphi \gamma_{i,j})$ for neighboring Majorana modes $i,j$) does not in general correspond to locality with respect to the structure of the qubit lattice.} 
\end{figure} 

 In this part we discuss the gate depths and numbers required to implement our joint measurement strategy on a quantum computer in different fermion-to-qubit mappings and qubit layouts.
 
\emph{General strategy. } In standard fermion-to-qubit mappings (like Jordan-Wigner \cite{jordan93}, Parity \cite{Seeley_2012}, Bravyi-Kitaev \cite{bravyi02}, or superfast encodings \cite{bravyi02, Setia_2019}; see \cite{Seeley_2012} for a review), unitaries from step (i) of Sec. \ref{sec:Genprotocol} can be implemented via a depth-1 circuit consisting of Pauli gates. On the other hand, the realization of any $U\in\mathrm{FLO}(N)$, including those from step (ii), can be achieved with a circuit of depth at most $d_{\mathrm{gen}}=\mathcal{O}(N^3)$ and $n_{\mathrm{gen}}=\mathcal{O}(N^4)$ two-qubit gates in any of the encodings. This follows from two observations. First, an arbitrary $U\in\mathrm{FLO}(N)$ can be realized by $U=\tilde{U} \gamma_1^{s}$,   where $\tilde{U}$ is parity preserving (i.e., it commutes with $Q:=\gamma_{[2N]}$), and $\gamma_1^{s}$ ($s=0,1$) is responsible for the (optional) change of the parity. It is known that $\tilde{U}$ can be implemented using   $\mathcal{O}(N^2)$ gates of the form $\exp(\mathrm{i}\varphi \gamma_{i,i+1})$ \cite{Boxio18,oszmaniec22}. Second, each of the unitaries $\exp(\mathrm{i} \varphi \gamma_{ij})$ can be realized in depth $\mathcal{O}(N)$, which is a consequence of the fermion-to-qubit encoding that maps $\gamma_{ij}$ to a Pauli string with support on at most $\mathcal{O}(N)$ qubits, and is thus realizable in depth $\mathcal{O}(N)$ and using $\mathcal{O}(N)$ gates (unitaries $\exp(\mathrm{i} \varphi P )$, for a Pauli string $P$ supported on $k$ qubits can be implemented via $\mathcal{O}(k)$ two-qubit gates \cite{nielsen00}).

However, when using the standard Jordan-Wigner encoding on a line (see Fig. \ref{fig:physLayouts}a) it is possible to realize any $U\in\mathrm{FLO}(N)$ using $\mathcal{O}(N)$ layers of   $\exp(\mathrm{i} \varphi \gamma_{i,i+1})$ applied in parallel and forming a brickwork pattern on a triangle layout \cite{oszmaniec22}. Furthermore, in this encoding $\exp(\mathrm{i} \varphi \gamma_{i,i+1})$ can be realized by a constant depth unitary acting on neighboring qubits. Rotations corresponding to disjoint Majorana pair operators belonging to each layer can be realized in parallel (see \cite{Boxio18} or \cite{oszmaniec22}). This yields the overall depth $d_{\mathrm{JW}}=\mathcal{O}(N)$ and the two-qubit gate-count $n_{\mathrm{JW}}=\mathcal{O}(N^2)$.

\emph{A rectangular qubit grid in the Jordan-Wigner encoding.} Due to the particular structure we impose on the unitaries in step (ii), our joint measurement scheme for Majorana pairs and quadruples can be realized with a smaller circuit depth and using fewer two-qubit gates in a 2D layout of qubits.  Consider a $L/2 \times (L+1)$ lattice of qubits realizing $L(L+1)$ Majorana modes via the Jordan-Wigner encoding (see Fig. \ref{fig:physLayouts}b). In Sec. \ref{sec:constructionOFunitaries} we described specific FLO unitaries $U_{\mathrm{pair}}, U_{\mathrm{sup}}, U^{(r)}_{\mathrm{resh}}$ that define measurement rounds of the protocol from Prop. \ref{prop:Randomized}. Importantly, for the variant of the Jordan-Wigner encoding in which the Majorana modes are labelled as in Fig.  \ref{fig:physLayouts}b, $U_{\mathrm{sup}}$, $U_{\mathrm{pair}}$ and $U^{(r)}_{\mathrm{resh}}$ can all be implemented in depth $d_{2D} = \mathcal{O}(N^{1/2})$ and with $n_{2D}= \mathcal{O}(N^{3/2})$ two-qubit gates. In what follows we sketch the reasoning behind these estimates.

First, the superposition unitary $U_{\mathrm{sup}}$ is \emph{horizontal} (see \eqref{eq:vertical_horizontalO}) in the 2D Majorana layout in Fig. \ref{fig:LL+1} (and consequently Fig. \ref{fig:physLayouts}b). Such transformations are consistent with the ordering that gives rise to the JW encoding and can therefore be implemented by a (tensor) product of $L+1$ independent FLO unitaries, each acting on qubits arranged in a given row. Therefore $U_{\mathrm{sup}}$ can be implemented in depth $\mathcal{O}(N^{1/2})$ and $\mathcal{O}(N^{3/2})$ nearest-neighbor two-qubit gates.

It is more challenging to implement the  unitaries $U_{\mathrm{pair}}$ and $U_{\mathrm{resh}}^{(r)}$. Note that, as explained in Sec. \ref{sec:constructionOFunitaries}, they are purely \textit{vertical} (see \eqref{eq:vertical_horizontalO}). That is, the permutations $\pi$ and $\sigma^{(r)}$ that correspond to $U_{\mathrm{pair}}$ and $U^{(r)}_{\mathrm{resh}}$ respectively, correspond to orthogonal transformations $R_{\pi/\sigma^{(r)}}=\oplus_{\beta=1}^{L/2} \tau^{\beta}$, where $\tau^{\beta}$ only permutes elements in the $\beta$--th column of qubits in the 2D layout under JW encoding. Importantly, in \cite{Boxio18} (see also \cite{Montanaro2020}), it was shown that such a class of FLO gates can be implemented in depth $\mathcal{O}(L)=\mathcal{O}(N^{1/2})$ and using $\mathcal{O}(L^3)=\mathcal{O}(N^{3/2})$ two-qubit gates, by utilizing one extra auxiliary qubit per row. In Supplementary Information \ref{app:decomp} we provide the proof that these scalings hold for general FLO transformations, extending the treatment of number preserving FLO transformations in \cite{Boxio18}). Thus, $U_{\mathrm{pair}}$ and $U^{(r)}_{\mathrm{resh}}$ can also be implemented with the claimed gate depth and count. \rev{Numerical experiments in \cite{Montanaro2020}, which consider the complexity of the method for passive FLO transformations, adapted from \cite{Boxio18}, suggest that the gate scaling advantage becomes apparent for systems larger than 18$\times$18 qubits due to the gate overhead incurred by the method of implementing the vertical transformations. We therefore expect, in practice, that our method exhibits advantageous gate scaling in the 2D setting for system sizes of the order of hundreds of qubits.
}

\section*{Data Availability}

The atomic coordinates of the molecular Hamiltonians used for numerical benchmarks are available in the Supplementary Information \ref{app:numerics}. The code utilized for the benchmarks is available upon request.

\section*{Acknowledgements}
 We would like to thank Ingo Roth, Martin Kliesch and Alexander Gresch for helpful discussions, as well as Katarzyna Kowalczyk-Murynka for assistance with numerical computations. The authors acknowledge financial support from the TEAM-NET project co-financed by the EU within the Smart Growth Operational Program (contract no. POIR.04.04.00-00-17C1/18-00). D.M. acknowledges support from PNRR MUR Project No. PE0000023-NQSTI. M.O. acknowledges support from QuantERA Project TouQan financed by the National Science Center No. 2023/05/Y/ST2/00140 (ERA-NET Cofund QuantERA II).

  \section*{Author contributions}  M.O. developed the initial research idea and oversaw all aspects of the project. D.M. and M.O. were involved in the construction of the joint measurement and its application to estimating Majorana pairs, quadruples and Hamiltonians. J.M., D.M. and M.O. participated in developing concrete sample efficient schemes with favourable gate depths and counts. J.M. conducted all numerical calculations. All authors contributed to the discussions, analytical calculations, and writing of the manuscript.
 
\flushcolsend

\newpage 

\bibliographystyle{apsrev4-2}

\bibliography{refs}

\begin{thebibliography}{58}%
\makeatletter
\providecommand \@ifxundefined [1]{%
 \@ifx{#1\undefined}
}%
\providecommand \@ifnum [1]{%
 \ifnum #1\expandafter \@firstoftwo
 \else \expandafter \@secondoftwo
 \fi
}%
\providecommand \@ifx [1]{%
 \ifx #1\expandafter \@firstoftwo
 \else \expandafter \@secondoftwo
 \fi
}%
\providecommand \natexlab [1]{#1}%
\providecommand \enquote  [1]{``#1''}%
\providecommand \bibnamefont  [1]{#1}%
\providecommand \bibfnamefont [1]{#1}%
\providecommand \citenamefont [1]{#1}%
\providecommand \href@noop [0]{\@secondoftwo}%
\providecommand \href [0]{\begingroup \@sanitize@url \@href}%
\providecommand \@href[1]{\@@startlink{#1}\@@href}%
\providecommand \@@href[1]{\endgroup#1\@@endlink}%
\providecommand \@sanitize@url [0]{\catcode `\\12\catcode `\$12\catcode `\&12\catcode `\#12\catcode `\^12\catcode `\_12\catcode `\%12\relax}%
\providecommand \@@startlink[1]{}%
\providecommand \@@endlink[0]{}%
\providecommand \url  [0]{\begingroup\@sanitize@url \@url }%
\providecommand \@url [1]{\endgroup\@href {#1}{\urlprefix }}%
\providecommand \urlprefix  [0]{URL }%
\providecommand \Eprint [0]{\href }%
\providecommand \doibase [0]{https://doi.org/}%
\providecommand \selectlanguage [0]{\@gobble}%
\providecommand \bibinfo  [0]{\@secondoftwo}%
\providecommand \bibfield  [0]{\@secondoftwo}%
\providecommand \translation [1]{[#1]}%
\providecommand \BibitemOpen [0]{}%
\providecommand \bibitemStop [0]{}%
\providecommand \bibitemNoStop [0]{.\EOS\space}%
\providecommand \EOS [0]{\spacefactor3000\relax}%
\providecommand \BibitemShut  [1]{\csname bibitem#1\endcsname}%
\let\auto@bib@innerbib\@empty
\bibitem [{\citenamefont {Zhao}\ \emph {et~al.}(2021)\citenamefont {Zhao}, \citenamefont {Rubin},\ and\ \citenamefont {Miyake}}]{Zhao21}%
  \BibitemOpen
  \bibfield  {author} {\bibinfo {author} {\bibfnamefont {A.}~\bibnamefont {Zhao}}, \bibinfo {author} {\bibfnamefont {N.~C.}\ \bibnamefont {Rubin}},\ and\ \bibinfo {author} {\bibfnamefont {A.}~\bibnamefont {Miyake}},\ }\href {https://doi.org/10.1103/PhysRevLett.127.110504} {\bibfield  {journal} {\bibinfo  {journal} {Phys. Rev. Lett.}\ }\textbf {\bibinfo {volume} {127}},\ \bibinfo {pages} {0110504} (\bibinfo {year} {2021})}\BibitemShut {NoStop}%
\bibitem [{\citenamefont {{Wan}}\ \emph {et~al.}(2023)\citenamefont {{Wan}}, \citenamefont {{Huggins}}, \citenamefont {{Lee}},\ and\ \citenamefont {{Babbush}}}]{Wan22}%
  \BibitemOpen
  \bibfield  {author} {\bibinfo {author} {\bibfnamefont {K.}~\bibnamefont {{Wan}}}, \bibinfo {author} {\bibfnamefont {W.~J.}\ \bibnamefont {{Huggins}}}, \bibinfo {author} {\bibfnamefont {J.}~\bibnamefont {{Lee}}},\ and\ \bibinfo {author} {\bibfnamefont {R.}~\bibnamefont {{Babbush}}},\ }\href {https://doi.org/10.1007/s00220-023-04844-0} {\bibfield  {journal} {\bibinfo  {journal} {Comm. Math. Phys.}\ }\textbf {\bibinfo {volume} {404}},\ \bibinfo {pages} {629} (\bibinfo {year} {2023})}\BibitemShut {NoStop}%
\bibitem [{\citenamefont {Bauer}\ \emph {et~al.}(2020)\citenamefont {Bauer}, \citenamefont {Bravyi}, \citenamefont {Motta},\ and\ \citenamefont {Chan}}]{bauer20}%
  \BibitemOpen
  \bibfield  {author} {\bibinfo {author} {\bibfnamefont {B.}~\bibnamefont {Bauer}}, \bibinfo {author} {\bibfnamefont {S.}~\bibnamefont {Bravyi}}, \bibinfo {author} {\bibfnamefont {M.}~\bibnamefont {Motta}},\ and\ \bibinfo {author} {\bibfnamefont {G.~K.-L.}\ \bibnamefont {Chan}},\ }\href {https://doi.org/https://doi.org/10.1021/acs.chemrev.9b00829} {\bibfield  {journal} {\bibinfo  {journal} {Chem. Rev.}\ }\textbf {\bibinfo {volume} {120}},\ \bibinfo {pages} {12685} (\bibinfo {year} {2020})}\BibitemShut {NoStop}%
\bibitem [{\citenamefont {Bharti}\ \emph {et~al.}(2022)\citenamefont {Bharti}, \citenamefont {Cervera-Lierta}, \citenamefont {Kyaw}, \citenamefont {Haug}, \citenamefont {Alperin-Lea}, \citenamefont {Anand}, \citenamefont {Degroote}, \citenamefont {Heimonen}, \citenamefont {Kottmann}, \citenamefont {Menke}, \citenamefont {Mok}, \citenamefont {Sim}, \citenamefont {Kwek},\ and\ \citenamefont {Aspuru-Guzik}}]{bharti22}%
  \BibitemOpen
  \bibfield  {author} {\bibinfo {author} {\bibfnamefont {K.}~\bibnamefont {Bharti}}, \bibinfo {author} {\bibfnamefont {A.}~\bibnamefont {Cervera-Lierta}}, \bibinfo {author} {\bibfnamefont {T.~H.}\ \bibnamefont {Kyaw}}, \bibinfo {author} {\bibfnamefont {T.}~\bibnamefont {Haug}}, \bibinfo {author} {\bibfnamefont {S.}~\bibnamefont {Alperin-Lea}}, \bibinfo {author} {\bibfnamefont {A.}~\bibnamefont {Anand}}, \bibinfo {author} {\bibfnamefont {M.}~\bibnamefont {Degroote}}, \bibinfo {author} {\bibfnamefont {H.}~\bibnamefont {Heimonen}}, \bibinfo {author} {\bibfnamefont {J.~S.}\ \bibnamefont {Kottmann}}, \bibinfo {author} {\bibfnamefont {T.}~\bibnamefont {Menke}}, \bibinfo {author} {\bibfnamefont {W.~K.}\ \bibnamefont {Mok}}, \bibinfo {author} {\bibfnamefont {S.}~\bibnamefont {Sim}}, \bibinfo {author} {\bibfnamefont {L.~C.}\ \bibnamefont {Kwek}},\ and\ \bibinfo {author} {\bibfnamefont {A.}~\bibnamefont {Aspuru-Guzik}},\ }\href {https://doi.org/10.1103/RevModPhys.94.015004} {\bibfield  {journal} {\bibinfo  {journal}
  {Rev. Mod. Phys.}\ }\textbf {\bibinfo {volume} {94}},\ \bibinfo {pages} {015004} (\bibinfo {year} {2022})}\BibitemShut {NoStop}%
\bibitem [{\citenamefont {Preskill}(2018)}]{preskill18}%
  \BibitemOpen
  \bibfield  {author} {\bibinfo {author} {\bibfnamefont {J.}~\bibnamefont {Preskill}},\ }\href {https://doi.org/10.22331/q-2018-08-06-79} {\bibfield  {journal} {\bibinfo  {journal} {Quantum}\ }\textbf {\bibinfo {volume} {2}},\ \bibinfo {pages} {79} (\bibinfo {year} {2018})}\BibitemShut {NoStop}%
\bibitem [{\citenamefont {McClean}\ \emph {et~al.}(2016)\citenamefont {McClean}, \citenamefont {Romero}, \citenamefont {Babbush},\ and\ \citenamefont {Aspuru-Guzik}}]{mcclean16}%
  \BibitemOpen
  \bibfield  {author} {\bibinfo {author} {\bibfnamefont {J.~R.}\ \bibnamefont {McClean}}, \bibinfo {author} {\bibfnamefont {J.}~\bibnamefont {Romero}}, \bibinfo {author} {\bibfnamefont {R.}~\bibnamefont {Babbush}},\ and\ \bibinfo {author} {\bibfnamefont {A.}~\bibnamefont {Aspuru-Guzik}},\ }\href {https://doi.org/10.1088/1367-2630/18/2/023023} {\bibfield  {journal} {\bibinfo  {journal} {New J. Phys.}\ }\textbf {\bibinfo {volume} {18}},\ \bibinfo {pages} {023023} (\bibinfo {year} {2016})}\BibitemShut {NoStop}%
\bibitem [{\citenamefont {Garc{\'\i}a-P{\'e}rez}\ \emph {et~al.}(2021)\citenamefont {Garc{\'\i}a-P{\'e}rez}, \citenamefont {Rossi}, \citenamefont {Sokolov}, \citenamefont {Tacchino}, \citenamefont {Barkoutsos}, \citenamefont {Mazzola}, \citenamefont {Tavernelli},\ and\ \citenamefont {Maniscalco}}]{garcia21}%
  \BibitemOpen
  \bibfield  {author} {\bibinfo {author} {\bibfnamefont {G.}~\bibnamefont {Garc{\'\i}a-P{\'e}rez}}, \bibinfo {author} {\bibfnamefont {M.~A.}\ \bibnamefont {Rossi}}, \bibinfo {author} {\bibfnamefont {B.}~\bibnamefont {Sokolov}}, \bibinfo {author} {\bibfnamefont {F.}~\bibnamefont {Tacchino}}, \bibinfo {author} {\bibfnamefont {P.~K.}\ \bibnamefont {Barkoutsos}}, \bibinfo {author} {\bibfnamefont {G.}~\bibnamefont {Mazzola}}, \bibinfo {author} {\bibfnamefont {I.}~\bibnamefont {Tavernelli}},\ and\ \bibinfo {author} {\bibfnamefont {S.}~\bibnamefont {Maniscalco}},\ }\href {https://doi.org/https://doi.org/10.1103/PRXQuantum.2.040342} {\bibfield  {journal} {\bibinfo  {journal} {PRX Quantum}\ }\textbf {\bibinfo {volume} {2}},\ \bibinfo {pages} {040342} (\bibinfo {year} {2021})}\BibitemShut {NoStop}%
\bibitem [{\citenamefont {Miessen}\ \emph {et~al.}(2023)\citenamefont {Miessen}, \citenamefont {Ollitrault}, \citenamefont {Tacchino},\ and\ \citenamefont {Tavernelli}}]{miessen_quantum_2023}%
  \BibitemOpen
  \bibfield  {author} {\bibinfo {author} {\bibfnamefont {A.}~\bibnamefont {Miessen}}, \bibinfo {author} {\bibfnamefont {P.~J.}\ \bibnamefont {Ollitrault}}, \bibinfo {author} {\bibfnamefont {F.}~\bibnamefont {Tacchino}},\ and\ \bibinfo {author} {\bibfnamefont {I.}~\bibnamefont {Tavernelli}},\ }\href {https://doi.org/10.1038/s43588-022-00374-2} {\bibfield  {journal} {\bibinfo  {journal} {Nat. Comput. Sci.}\ }\textbf {\bibinfo {volume} {3}},\ \bibinfo {pages} {25} (\bibinfo {year} {2023})}\BibitemShut {NoStop}%
\bibitem [{\citenamefont {Babbush}\ \emph {et~al.}(2023)\citenamefont {Babbush}, \citenamefont {Huggins}, \citenamefont {Berry}, \citenamefont {Ung}, \citenamefont {Zhao}, \citenamefont {Reichman}, \citenamefont {Neven}, \citenamefont {Baczewski},\ and\ \citenamefont {Lee}}]{babbush_quantum_2023}%
  \BibitemOpen
  \bibfield  {author} {\bibinfo {author} {\bibfnamefont {R.}~\bibnamefont {Babbush}}, \bibinfo {author} {\bibfnamefont {W.~J.}\ \bibnamefont {Huggins}}, \bibinfo {author} {\bibfnamefont {D.~W.}\ \bibnamefont {Berry}}, \bibinfo {author} {\bibfnamefont {S.~F.}\ \bibnamefont {Ung}}, \bibinfo {author} {\bibfnamefont {A.}~\bibnamefont {Zhao}}, \bibinfo {author} {\bibfnamefont {D.~R.}\ \bibnamefont {Reichman}}, \bibinfo {author} {\bibfnamefont {H.}~\bibnamefont {Neven}}, \bibinfo {author} {\bibfnamefont {A.~D.}\ \bibnamefont {Baczewski}},\ and\ \bibinfo {author} {\bibfnamefont {J.}~\bibnamefont {Lee}},\ }\href {https://doi.org/10.1038/s41467-023-39024-0} {\bibfield  {journal} {\bibinfo  {journal} {Nat. Commun.}\ }\textbf {\bibinfo {volume} {14}},\ \bibinfo {pages} {4058} (\bibinfo {year} {2023})}\BibitemShut {NoStop}%
\bibitem [{\citenamefont {Gidofalvi}\ and\ \citenamefont {Mazziotti}(2007)}]{gidofalvi07}%
  \BibitemOpen
  \bibfield  {author} {\bibinfo {author} {\bibfnamefont {G.}~\bibnamefont {Gidofalvi}}\ and\ \bibinfo {author} {\bibfnamefont {D.~A.}\ \bibnamefont {Mazziotti}},\ }\href {https://doi.org/https://doi.org/10.1063/1.2423008} {\bibfield  {journal} {\bibinfo  {journal} {J. Chem. Phys.}\ }\textbf {\bibinfo {volume} {126}},\ \bibinfo {pages} {024105} (\bibinfo {year} {2007})}\BibitemShut {NoStop}%
\bibitem [{\citenamefont {O’Brien}\ \emph {et~al.}(2019)\citenamefont {O’Brien}, \citenamefont {Senjean}, \citenamefont {Sagastizabal}, \citenamefont {Bonet-Monroig}, \citenamefont {Dutkiewicz}, \citenamefont {Buda}, \citenamefont {DiCarlo},\ and\ \citenamefont {Visscher}}]{obrien19}%
  \BibitemOpen
  \bibfield  {author} {\bibinfo {author} {\bibfnamefont {T.~E.}\ \bibnamefont {O’Brien}}, \bibinfo {author} {\bibfnamefont {B.}~\bibnamefont {Senjean}}, \bibinfo {author} {\bibfnamefont {R.}~\bibnamefont {Sagastizabal}}, \bibinfo {author} {\bibfnamefont {X.}~\bibnamefont {Bonet-Monroig}}, \bibinfo {author} {\bibfnamefont {A.}~\bibnamefont {Dutkiewicz}}, \bibinfo {author} {\bibfnamefont {F.}~\bibnamefont {Buda}}, \bibinfo {author} {\bibfnamefont {L.}~\bibnamefont {DiCarlo}},\ and\ \bibinfo {author} {\bibfnamefont {L.}~\bibnamefont {Visscher}},\ }\href {https://doi.org/https://doi.org/10.1038/s41534-019-0213-4} {\bibfield  {journal} {\bibinfo  {journal} {npj Quant. Inf.}\ }\textbf {\bibinfo {volume} {5}},\ \bibinfo {pages} {113} (\bibinfo {year} {2019})}\BibitemShut {NoStop}%
\bibitem [{\citenamefont {Overy}\ \emph {et~al.}(2014)\citenamefont {Overy}, \citenamefont {Booth}, \citenamefont {Blunt}, \citenamefont {Shepherd}, \citenamefont {Cleland},\ and\ \citenamefont {Alavi}}]{overy14}%
  \BibitemOpen
  \bibfield  {author} {\bibinfo {author} {\bibfnamefont {C.}~\bibnamefont {Overy}}, \bibinfo {author} {\bibfnamefont {G.~H.}\ \bibnamefont {Booth}}, \bibinfo {author} {\bibfnamefont {N.}~\bibnamefont {Blunt}}, \bibinfo {author} {\bibfnamefont {J.~J.}\ \bibnamefont {Shepherd}}, \bibinfo {author} {\bibfnamefont {D.}~\bibnamefont {Cleland}},\ and\ \bibinfo {author} {\bibfnamefont {A.}~\bibnamefont {Alavi}},\ }\href {https://doi.org/https://doi.org/10.1063/1.4904313} {\bibfield  {journal} {\bibinfo  {journal} {J. Chem. Phys.}\ }\textbf {\bibinfo {volume} {141}},\ \bibinfo {pages} {244117} (\bibinfo {year} {2014})}\BibitemShut {NoStop}%
\bibitem [{\citenamefont {McClean}\ \emph {et~al.}(2017)\citenamefont {McClean}, \citenamefont {Kimchi-Schwartz}, \citenamefont {Carter},\ and\ \citenamefont {De~Jong}}]{mcclean17}%
  \BibitemOpen
  \bibfield  {author} {\bibinfo {author} {\bibfnamefont {J.~R.}\ \bibnamefont {McClean}}, \bibinfo {author} {\bibfnamefont {M.~E.}\ \bibnamefont {Kimchi-Schwartz}}, \bibinfo {author} {\bibfnamefont {J.}~\bibnamefont {Carter}},\ and\ \bibinfo {author} {\bibfnamefont {W.~A.}\ \bibnamefont {De~Jong}},\ }\href {https://doi.org/https://doi.org/10.1103/PhysRevA.95.042308} {\bibfield  {journal} {\bibinfo  {journal} {Phys. Rev. A}\ }\textbf {\bibinfo {volume} {95}},\ \bibinfo {pages} {042308} (\bibinfo {year} {2017})}\BibitemShut {NoStop}%
\bibitem [{\citenamefont {Takeshita}\ \emph {et~al.}(2020)\citenamefont {Takeshita}, \citenamefont {Rubin}, \citenamefont {Jiang}, \citenamefont {Lee}, \citenamefont {Babbush},\ and\ \citenamefont {McClean}}]{takeshita20}%
  \BibitemOpen
  \bibfield  {author} {\bibinfo {author} {\bibfnamefont {T.}~\bibnamefont {Takeshita}}, \bibinfo {author} {\bibfnamefont {N.~C.}\ \bibnamefont {Rubin}}, \bibinfo {author} {\bibfnamefont {Z.}~\bibnamefont {Jiang}}, \bibinfo {author} {\bibfnamefont {E.}~\bibnamefont {Lee}}, \bibinfo {author} {\bibfnamefont {R.}~\bibnamefont {Babbush}},\ and\ \bibinfo {author} {\bibfnamefont {J.~R.}\ \bibnamefont {McClean}},\ }\href {https://doi.org/https://doi.org/10.1103/PhysRevX.10.011004} {\bibfield  {journal} {\bibinfo  {journal} {Phys. Rev. X}\ }\textbf {\bibinfo {volume} {10}},\ \bibinfo {pages} {011004} (\bibinfo {year} {2020})}\BibitemShut {NoStop}%
\bibitem [{\citenamefont {Gluza}\ \emph {et~al.}(2018)\citenamefont {Gluza}, \citenamefont {Kliesch}, \citenamefont {Eisert},\ and\ \citenamefont {Aolita}}]{gluza18}%
  \BibitemOpen
  \bibfield  {author} {\bibinfo {author} {\bibfnamefont {M.}~\bibnamefont {Gluza}}, \bibinfo {author} {\bibfnamefont {M.}~\bibnamefont {Kliesch}}, \bibinfo {author} {\bibfnamefont {J.}~\bibnamefont {Eisert}},\ and\ \bibinfo {author} {\bibfnamefont {L.}~\bibnamefont {Aolita}},\ }\href {https://doi.org/https://doi.org/10.1103/PhysRevLett.120.190501} {\bibfield  {journal} {\bibinfo  {journal} {Phys. Rev. Lett.}\ }\textbf {\bibinfo {volume} {120}},\ \bibinfo {pages} {190501} (\bibinfo {year} {2018})}\BibitemShut {NoStop}%
\bibitem [{\citenamefont {Jena}\ \emph {et~al.}(2019)\citenamefont {Jena}, \citenamefont {Genin},\ and\ \citenamefont {Mosca}}]{jena19}%
  \BibitemOpen
  \bibfield  {author} {\bibinfo {author} {\bibfnamefont {A.}~\bibnamefont {Jena}}, \bibinfo {author} {\bibfnamefont {S.}~\bibnamefont {Genin}},\ and\ \bibinfo {author} {\bibfnamefont {M.}~\bibnamefont {Mosca}},\ }\href {https://doi.org/https://doi.org/10.48550/arXiv.1907.07859} {\bibfield  {journal} {\bibinfo  {journal} {arXiv preprint arXiv:1907.07859}\ } (\bibinfo {year} {2019})}\BibitemShut {NoStop}%
\bibitem [{\citenamefont {Gokhale}\ \emph {et~al.}(2019)\citenamefont {Gokhale}, \citenamefont {Angiuli}, \citenamefont {Ding}, \citenamefont {Gui}, \citenamefont {Tomesh}, \citenamefont {Suchara}, \citenamefont {Martonosi},\ and\ \citenamefont {Chong}}]{gokhale19}%
  \BibitemOpen
  \bibfield  {author} {\bibinfo {author} {\bibfnamefont {P.}~\bibnamefont {Gokhale}}, \bibinfo {author} {\bibfnamefont {O.}~\bibnamefont {Angiuli}}, \bibinfo {author} {\bibfnamefont {Y.}~\bibnamefont {Ding}}, \bibinfo {author} {\bibfnamefont {K.}~\bibnamefont {Gui}}, \bibinfo {author} {\bibfnamefont {T.}~\bibnamefont {Tomesh}}, \bibinfo {author} {\bibfnamefont {M.}~\bibnamefont {Suchara}}, \bibinfo {author} {\bibfnamefont {M.}~\bibnamefont {Martonosi}},\ and\ \bibinfo {author} {\bibfnamefont {F.~T.}\ \bibnamefont {Chong}},\ }\href {https://doi.org/https://doi.org/10.48550/arXiv.1907.13623} {\bibfield  {journal} {\bibinfo  {journal} {arXiv preprint arXiv:1907.13623}\ } (\bibinfo {year} {2019})}\BibitemShut {NoStop}%
\bibitem [{\citenamefont {Yen}\ \emph {et~al.}(2020)\citenamefont {Yen}, \citenamefont {Verteletskyi},\ and\ \citenamefont {Izmaylov}}]{yen20}%
  \BibitemOpen
  \bibfield  {author} {\bibinfo {author} {\bibfnamefont {T.~C.}\ \bibnamefont {Yen}}, \bibinfo {author} {\bibfnamefont {V.}~\bibnamefont {Verteletskyi}},\ and\ \bibinfo {author} {\bibfnamefont {A.~F.}\ \bibnamefont {Izmaylov}},\ }\href {https://doi.org/10.1021/acs.jctc.0c00008} {\bibfield  {journal} {\bibinfo  {journal} {J. Chem. Theory Comput.}\ }\textbf {\bibinfo {volume} {16}},\ \bibinfo {pages} {2400} (\bibinfo {year} {2020})}\BibitemShut {NoStop}%
\bibitem [{\citenamefont {Verteletskyi}\ \emph {et~al.}(2020)\citenamefont {Verteletskyi}, \citenamefont {Yen},\ and\ \citenamefont {Izmaylov}}]{verteletskyi20}%
  \BibitemOpen
  \bibfield  {author} {\bibinfo {author} {\bibfnamefont {V.}~\bibnamefont {Verteletskyi}}, \bibinfo {author} {\bibfnamefont {T.~C.}\ \bibnamefont {Yen}},\ and\ \bibinfo {author} {\bibfnamefont {A.~F.}\ \bibnamefont {Izmaylov}},\ }\href {https://doi.org/10.1063/1.5141458} {\bibfield  {journal} {\bibinfo  {journal} {J. Chem. Phys.}\ }\textbf {\bibinfo {volume} {152}},\ \bibinfo {pages} {124114} (\bibinfo {year} {2020})}\BibitemShut {NoStop}%
\bibitem [{\citenamefont {Crawford}\ \emph {et~al.}(2021)\citenamefont {Crawford}, \citenamefont {Straaten}, \citenamefont {Wang}, \citenamefont {Parks}, \citenamefont {Campbell},\ and\ \citenamefont {Brierley}}]{crawford21}%
  \BibitemOpen
  \bibfield  {author} {\bibinfo {author} {\bibfnamefont {O.}~\bibnamefont {Crawford}}, \bibinfo {author} {\bibfnamefont {B.~V.}\ \bibnamefont {Straaten}}, \bibinfo {author} {\bibfnamefont {D.}~\bibnamefont {Wang}}, \bibinfo {author} {\bibfnamefont {T.}~\bibnamefont {Parks}}, \bibinfo {author} {\bibfnamefont {E.}~\bibnamefont {Campbell}},\ and\ \bibinfo {author} {\bibfnamefont {S.}~\bibnamefont {Brierley}},\ }\href {https://doi.org/10.22331/Q-2021-01-20-385} {\bibfield  {journal} {\bibinfo  {journal} {Quantum}\ }\textbf {\bibinfo {volume} {5}},\ \bibinfo {pages} {385} (\bibinfo {year} {2021})}\BibitemShut {NoStop}%
\bibitem [{\citenamefont {Izmaylov}\ \emph {et~al.}(2020)\citenamefont {Izmaylov}, \citenamefont {Yen}, \citenamefont {Lang},\ and\ \citenamefont {Verteletskyi}}]{izmaylov20}%
  \BibitemOpen
  \bibfield  {author} {\bibinfo {author} {\bibfnamefont {A.~F.}\ \bibnamefont {Izmaylov}}, \bibinfo {author} {\bibfnamefont {T.~C.}\ \bibnamefont {Yen}}, \bibinfo {author} {\bibfnamefont {R.~A.}\ \bibnamefont {Lang}},\ and\ \bibinfo {author} {\bibfnamefont {V.}~\bibnamefont {Verteletskyi}},\ }\href {https://doi.org/10.1021/acs.jctc.9b00791} {\bibfield  {journal} {\bibinfo  {journal} {J. Chem. Theory Comput.}\ }\textbf {\bibinfo {volume} {16}},\ \bibinfo {pages} {190} (\bibinfo {year} {2020})}\BibitemShut {NoStop}%
\bibitem [{\citenamefont {Zhao}\ \emph {et~al.}(2020)\citenamefont {Zhao}, \citenamefont {Tranter}, \citenamefont {Kirby}, \citenamefont {Ung}, \citenamefont {Miyake},\ and\ \citenamefont {Love}}]{Zhao20}%
  \BibitemOpen
  \bibfield  {author} {\bibinfo {author} {\bibfnamefont {A.}~\bibnamefont {Zhao}}, \bibinfo {author} {\bibfnamefont {A.}~\bibnamefont {Tranter}}, \bibinfo {author} {\bibfnamefont {W.~M.}\ \bibnamefont {Kirby}}, \bibinfo {author} {\bibfnamefont {S.~F.}\ \bibnamefont {Ung}}, \bibinfo {author} {\bibfnamefont {A.}~\bibnamefont {Miyake}},\ and\ \bibinfo {author} {\bibfnamefont {P.~J.}\ \bibnamefont {Love}},\ }\href {https://doi.org/10.1103/PhysRevA.101.062322} {\bibfield  {journal} {\bibinfo  {journal} {Phys. Rev. A}\ }\textbf {\bibinfo {volume} {101}},\ \bibinfo {pages} {062322} (\bibinfo {year} {2020})}\BibitemShut {NoStop}%
\bibitem [{\citenamefont {Bonet-Monroig}\ \emph {et~al.}(2020)\citenamefont {Bonet-Monroig}, \citenamefont {Babbush},\ and\ \citenamefont {O'Brien}}]{bonet20}%
  \BibitemOpen
  \bibfield  {author} {\bibinfo {author} {\bibfnamefont {X.}~\bibnamefont {Bonet-Monroig}}, \bibinfo {author} {\bibfnamefont {R.}~\bibnamefont {Babbush}},\ and\ \bibinfo {author} {\bibfnamefont {T.~E.}\ \bibnamefont {O'Brien}},\ }\href {https://doi.org/10.1103/PhysRevX.10.031064} {\bibfield  {journal} {\bibinfo  {journal} {Phys. Rev. X}\ }\textbf {\bibinfo {volume} {10}},\ \bibinfo {pages} {031064} (\bibinfo {year} {2020})}\BibitemShut {NoStop}%
\bibitem [{\citenamefont {Huang}\ \emph {et~al.}(2020)\citenamefont {Huang}, \citenamefont {Kueng},\ and\ \citenamefont {Preskill}}]{huang20}%
  \BibitemOpen
  \bibfield  {author} {\bibinfo {author} {\bibfnamefont {H.~Y.}\ \bibnamefont {Huang}}, \bibinfo {author} {\bibfnamefont {R.}~\bibnamefont {Kueng}},\ and\ \bibinfo {author} {\bibfnamefont {J.}~\bibnamefont {Preskill}},\ }\href {https://doi.org/10.1038/s41567-020-0932-7} {\bibfield  {journal} {\bibinfo  {journal} {Nat. Phys.}\ }\textbf {\bibinfo {volume} {16}},\ \bibinfo {pages} {1050} (\bibinfo {year} {2020})}\BibitemShut {NoStop}%
\bibitem [{\citenamefont {{Hadfield}}\ \emph {et~al.}(2022)\citenamefont {{Hadfield}}, \citenamefont {{Bravyi}}, \citenamefont {{Raymond}},\ and\ \citenamefont {{Mezzacapo}}}]{hadfield20}%
  \BibitemOpen
  \bibfield  {author} {\bibinfo {author} {\bibfnamefont {C.}~\bibnamefont {{Hadfield}}}, \bibinfo {author} {\bibfnamefont {S.}~\bibnamefont {{Bravyi}}}, \bibinfo {author} {\bibfnamefont {R.}~\bibnamefont {{Raymond}}},\ and\ \bibinfo {author} {\bibfnamefont {A.}~\bibnamefont {{Mezzacapo}}},\ }\href {https://doi.org/10.1007/s00220-022-04343-8} {\bibfield  {journal} {\bibinfo  {journal} {Comm. Math. Phys.}\ }\textbf {\bibinfo {volume} {391}},\ \bibinfo {pages} {951} (\bibinfo {year} {2022})}\BibitemShut {NoStop}%
\bibitem [{\citenamefont {Hu}\ \emph {et~al.}(2023)\citenamefont {Hu}, \citenamefont {Choi},\ and\ \citenamefont {You}}]{hu23}%
  \BibitemOpen
  \bibfield  {author} {\bibinfo {author} {\bibfnamefont {H.-Y.}\ \bibnamefont {Hu}}, \bibinfo {author} {\bibfnamefont {S.}~\bibnamefont {Choi}},\ and\ \bibinfo {author} {\bibfnamefont {Y.-Z.}\ \bibnamefont {You}},\ }\href {https://doi.org/https://doi.org/10.1103/PhysRevResearch.5.023027} {\bibfield  {journal} {\bibinfo  {journal} {Phys. Rev. Res.}\ }\textbf {\bibinfo {volume} {5}},\ \bibinfo {pages} {023027} (\bibinfo {year} {2023})}\BibitemShut {NoStop}%
\bibitem [{\citenamefont {Low}(2022)}]{Low22}%
  \BibitemOpen
  \bibfield  {author} {\bibinfo {author} {\bibfnamefont {G.~H.}\ \bibnamefont {Low}},\ }\href {https://doi.org/https://doi.org/10.48550/arXiv.2208.08964} {\bibfield  {journal} {\bibinfo  {journal} {arXiv preprint arXiv:2208.08964}\ } (\bibinfo {year} {2022})}\BibitemShut {NoStop}%
\bibitem [{\citenamefont {O'Gorman}(2022)}]{ogorman22}%
  \BibitemOpen
  \bibfield  {author} {\bibinfo {author} {\bibfnamefont {B.}~\bibnamefont {O'Gorman}},\ }\href {https://doi.org/https://doi.org/10.48550/arXiv.2207.14787} {\bibfield  {journal} {\bibinfo  {journal} {arXiv preprint arXiv:2207.14787}\ } (\bibinfo {year} {2022})}\BibitemShut {NoStop}%
\bibitem [{\citenamefont {Clinton}\ \emph {et~al.}(2024)\citenamefont {Clinton}, \citenamefont {Cubitt}, \citenamefont {Flynn}, \citenamefont {Gambetta}, \citenamefont {Klassen}, \citenamefont {Montanaro}, \citenamefont {Piddock}, \citenamefont {Santos},\ and\ \citenamefont {Sheridan}}]{clinton_towards_2024}%
  \BibitemOpen
  \bibfield  {author} {\bibinfo {author} {\bibfnamefont {L.}~\bibnamefont {Clinton}}, \bibinfo {author} {\bibfnamefont {T.}~\bibnamefont {Cubitt}}, \bibinfo {author} {\bibfnamefont {B.}~\bibnamefont {Flynn}}, \bibinfo {author} {\bibfnamefont {F.~M.}\ \bibnamefont {Gambetta}}, \bibinfo {author} {\bibfnamefont {J.}~\bibnamefont {Klassen}}, \bibinfo {author} {\bibfnamefont {A.}~\bibnamefont {Montanaro}}, \bibinfo {author} {\bibfnamefont {S.}~\bibnamefont {Piddock}}, \bibinfo {author} {\bibfnamefont {R.~A.}\ \bibnamefont {Santos}},\ and\ \bibinfo {author} {\bibfnamefont {E.}~\bibnamefont {Sheridan}},\ }\href {https://doi.org/10.1038/s41467-023-43479-6} {\bibfield  {journal} {\bibinfo  {journal} {Nature Communications}\ }\textbf {\bibinfo {volume} {15}},\ \bibinfo {pages} {211} (\bibinfo {year} {2024})}\BibitemShut {NoStop}%
\bibitem [{\citenamefont {McNulty}\ \emph {et~al.}(2023)\citenamefont {McNulty}, \citenamefont {Maciejewski},\ and\ \citenamefont {Oszmaniec}}]{mcnulty22}%
  \BibitemOpen
  \bibfield  {author} {\bibinfo {author} {\bibfnamefont {D.}~\bibnamefont {McNulty}}, \bibinfo {author} {\bibfnamefont {F.~B.}\ \bibnamefont {Maciejewski}},\ and\ \bibinfo {author} {\bibfnamefont {M.}~\bibnamefont {Oszmaniec}},\ }\href {https://doi.org/https://doi.org/10.1103/PhysRevLett.130.100801} {\bibfield  {journal} {\bibinfo  {journal} {Phys. Rev. Lett.}\ }\textbf {\bibinfo {volume} {130}},\ \bibinfo {pages} {100801} (\bibinfo {year} {2023})}\BibitemShut {NoStop}%
\bibitem [{\citenamefont {Busch}\ \emph {et~al.}(2016)\citenamefont {Busch}, \citenamefont {Lahti}, \citenamefont {Pellonp{\"a}{\"a}},\ and\ \citenamefont {Ylinen}}]{busch}%
  \BibitemOpen
  \bibfield  {author} {\bibinfo {author} {\bibfnamefont {P.}~\bibnamefont {Busch}}, \bibinfo {author} {\bibfnamefont {P.}~\bibnamefont {Lahti}}, \bibinfo {author} {\bibfnamefont {J.-P.}\ \bibnamefont {Pellonp{\"a}{\"a}}},\ and\ \bibinfo {author} {\bibfnamefont {K.}~\bibnamefont {Ylinen}},\ }\href {https://doi.org/https://doi.org/10.1007/978-3-319-43389-9} {\emph {\bibinfo {title} {Quantum measurement}}},\ Vol.~\bibinfo {volume} {23}\ (\bibinfo  {publisher} {Springer},\ \bibinfo {year} {2016})\BibitemShut {NoStop}%
\bibitem [{\citenamefont {Heinosaari}\ \emph {et~al.}(2008)\citenamefont {Heinosaari}, \citenamefont {Reitzner},\ and\ \citenamefont {Stano}}]{heinosaari08}%
  \BibitemOpen
  \bibfield  {author} {\bibinfo {author} {\bibfnamefont {T.}~\bibnamefont {Heinosaari}}, \bibinfo {author} {\bibfnamefont {D.}~\bibnamefont {Reitzner}},\ and\ \bibinfo {author} {\bibfnamefont {P.}~\bibnamefont {Stano}},\ }\href {https://doi.org/10.1007/s10701-008-9256-7} {\bibfield  {journal} {\bibinfo  {journal} {Found. Phys.}\ }\textbf {\bibinfo {volume} {38}},\ \bibinfo {pages} {1133} (\bibinfo {year} {2008})}\BibitemShut {NoStop}%
\bibitem [{\citenamefont {Heinosaari}\ \emph {et~al.}(2015)\citenamefont {Heinosaari}, \citenamefont {Kiukas},\ and\ \citenamefont {Reitzner}}]{heinosaari15}%
  \BibitemOpen
  \bibfield  {author} {\bibinfo {author} {\bibfnamefont {T.}~\bibnamefont {Heinosaari}}, \bibinfo {author} {\bibfnamefont {J.}~\bibnamefont {Kiukas}},\ and\ \bibinfo {author} {\bibfnamefont {D.}~\bibnamefont {Reitzner}},\ }\href {https://doi.org/10.1103/PhysRevA.92.022115} {\bibfield  {journal} {\bibinfo  {journal} {Phys. Rev. A}\ }\textbf {\bibinfo {volume} {92}},\ \bibinfo {pages} {022115} (\bibinfo {year} {2015})}\BibitemShut {NoStop}%
\bibitem [{\citenamefont {Berg}\ \emph {et~al.}(2022)\citenamefont {Berg}, \citenamefont {Minev},\ and\ \citenamefont {Temme}}]{berg_2022}%
  \BibitemOpen
  \bibfield  {author} {\bibinfo {author} {\bibfnamefont {E.~v.~d.}\ \bibnamefont {Berg}}, \bibinfo {author} {\bibfnamefont {Z.~K.}\ \bibnamefont {Minev}},\ and\ \bibinfo {author} {\bibfnamefont {K.}~\bibnamefont {Temme}},\ }\href {https://doi.org/10.1103/PhysRevA.105.032620} {\bibfield  {journal} {\bibinfo  {journal} {Physical Review A}\ }\textbf {\bibinfo {volume} {105}},\ \bibinfo {pages} {032620} (\bibinfo {year} {2022})}\BibitemShut {NoStop}%
\bibitem [{\citenamefont {Huggins}\ \emph {et~al.}(2021)\citenamefont {Huggins}, \citenamefont {McClean}, \citenamefont {Rubin}, \citenamefont {Jiang}, \citenamefont {Wiebe}, \citenamefont {Whaley},\ and\ \citenamefont {Babbush}}]{Huggins21}%
  \BibitemOpen
  \bibfield  {author} {\bibinfo {author} {\bibfnamefont {W.~J.}\ \bibnamefont {Huggins}}, \bibinfo {author} {\bibfnamefont {J.~R.}\ \bibnamefont {McClean}}, \bibinfo {author} {\bibfnamefont {N.~C.}\ \bibnamefont {Rubin}}, \bibinfo {author} {\bibfnamefont {Z.}~\bibnamefont {Jiang}}, \bibinfo {author} {\bibfnamefont {N.}~\bibnamefont {Wiebe}}, \bibinfo {author} {\bibfnamefont {K.~B.}\ \bibnamefont {Whaley}},\ and\ \bibinfo {author} {\bibfnamefont {R.}~\bibnamefont {Babbush}},\ }\href {https://doi.org/10.1038/s41534-020-00341-7} {\bibfield  {journal} {\bibinfo  {journal} {npj Quantum Inf.}\ }\textbf {\bibinfo {volume} {7}},\ \bibinfo {pages} {1} (\bibinfo {year} {2021})}\BibitemShut {NoStop}%
\bibitem [{\citenamefont {Hetényi}\ and\ \citenamefont {Wootton}(2024)}]{hetenyi_creating_2024}%
  \BibitemOpen
  \bibfield  {author} {\bibinfo {author} {\bibfnamefont {B.}~\bibnamefont {Hetényi}}\ and\ \bibinfo {author} {\bibfnamefont {J.~R.}\ \bibnamefont {Wootton}},\ }\href {https://doi.org/10.1103/PRXQuantum.5.040334} {\bibfield  {journal} {\bibinfo  {journal} {PRX Quantum}\ }\textbf {\bibinfo {volume} {5}},\ \bibinfo {pages} {040334} (\bibinfo {year} {2024})},\ \bibinfo {note} {publisher: American Physical Society}\BibitemShut {NoStop}%
\bibitem [{\citenamefont {McNulty}\ \emph {et~al.}(2024)\citenamefont {McNulty}, \citenamefont {Calegari},\ and\ \citenamefont {Oszmaniec}}]{MCO2024}%
  \BibitemOpen
  \bibfield  {author} {\bibinfo {author} {\bibfnamefont {D.}~\bibnamefont {McNulty}}, \bibinfo {author} {\bibfnamefont {S.}~\bibnamefont {Calegari}},\ and\ \bibinfo {author} {\bibfnamefont {M.}~\bibnamefont {Oszmaniec}},\ }\href {https://doi.org/https://doi.org/10.48550/arXiv.2402.19349} {\bibfield  {journal} {\bibinfo  {journal} {arXiv preprint arXiv:2402.19349}\ } (\bibinfo {year} {2024})}\BibitemShut {NoStop}%
\bibitem [{\citenamefont {{Bravyi}}\ \emph {et~al.}(2010)\citenamefont {{Bravyi}}, \citenamefont {{Terhal}},\ and\ \citenamefont {{Leemhuis}}}]{BravyiTerhav2010}%
  \BibitemOpen
  \bibfield  {author} {\bibinfo {author} {\bibfnamefont {S.}~\bibnamefont {{Bravyi}}}, \bibinfo {author} {\bibfnamefont {B.~M.}\ \bibnamefont {{Terhal}}},\ and\ \bibinfo {author} {\bibfnamefont {B.}~\bibnamefont {{Leemhuis}}},\ }\href {https://doi.org/10.1088/1367-2630/12/8/083039} {\bibfield  {journal} {\bibinfo  {journal} {New J. Phys.}\ }\textbf {\bibinfo {volume} {12}},\ \bibinfo {eid} {083039} (\bibinfo {year} {2010})}\BibitemShut {NoStop}%
\bibitem [{\citenamefont {Terhal}\ and\ \citenamefont {DiVincenzo}(2002)}]{terhal02}%
  \BibitemOpen
  \bibfield  {author} {\bibinfo {author} {\bibfnamefont {B.~M.}\ \bibnamefont {Terhal}}\ and\ \bibinfo {author} {\bibfnamefont {D.~P.}\ \bibnamefont {DiVincenzo}},\ }\href {https://doi.org/https://doi.org/10.1103/PhysRevA.65.032325} {\bibfield  {journal} {\bibinfo  {journal} {Phys. Rev. A}\ }\textbf {\bibinfo {volume} {65}},\ \bibinfo {pages} {032325} (\bibinfo {year} {2002})}\BibitemShut {NoStop}%
\bibitem [{\citenamefont {Bravyi}\ and\ \citenamefont {Kitaev}(2002)}]{bravyi02}%
  \BibitemOpen
  \bibfield  {author} {\bibinfo {author} {\bibfnamefont {S.~B.}\ \bibnamefont {Bravyi}}\ and\ \bibinfo {author} {\bibfnamefont {A.~Y.}\ \bibnamefont {Kitaev}},\ }\href {https://doi.org/https://doi.org/10.1006/aphy.2002.6254} {\bibfield  {journal} {\bibinfo  {journal} {Ann. Phys.}\ }\textbf {\bibinfo {volume} {298}},\ \bibinfo {pages} {210} (\bibinfo {year} {2002})}\BibitemShut {NoStop}%
\bibitem [{\citenamefont {Jozsa}\ and\ \citenamefont {Miyake}(2008)}]{jozsa08}%
  \BibitemOpen
  \bibfield  {author} {\bibinfo {author} {\bibfnamefont {R.}~\bibnamefont {Jozsa}}\ and\ \bibinfo {author} {\bibfnamefont {A.}~\bibnamefont {Miyake}},\ }\href {https://doi.org/https://doi.org/10.1098/rspa.2008.0189} {\bibfield  {journal} {\bibinfo  {journal} {Proc. R. Soc. A: Math. Phys. Eng. Sci.}\ }\textbf {\bibinfo {volume} {464}},\ \bibinfo {pages} {3089} (\bibinfo {year} {2008})}\BibitemShut {NoStop}%
\bibitem [{\citenamefont {Jordan}\ and\ \citenamefont {Wigner}(1993)}]{jordan93}%
  \BibitemOpen
  \bibfield  {author} {\bibinfo {author} {\bibfnamefont {P.}~\bibnamefont {Jordan}}\ and\ \bibinfo {author} {\bibfnamefont {E.~P.}\ \bibnamefont {Wigner}},\ }\href@noop {} {\emph {\bibinfo {title} {{\"U}ber das paulische {\"a}quivalenzverbot}}}\ (\bibinfo  {publisher} {Springer},\ \bibinfo {year} {1993})\BibitemShut {NoStop}%
\bibitem [{\citenamefont {Hoeffding}(1963)}]{Hoeffding63}%
  \BibitemOpen
  \bibfield  {author} {\bibinfo {author} {\bibfnamefont {W.}~\bibnamefont {Hoeffding}},\ }\href {https://doi.org/10.1080/01621459.1963.10500830} {\bibfield  {journal} {\bibinfo  {journal} {J. Am. Stat. Assoc.}\ }\textbf {\bibinfo {volume} {58}},\ \bibinfo {pages} {13} (\bibinfo {year} {1963})}\BibitemShut {NoStop}%
\bibitem [{\citenamefont {Jaming}\ and\ \citenamefont {Matolcsi}(2015)}]{Jaming15}%
  \BibitemOpen
  \bibfield  {author} {\bibinfo {author} {\bibfnamefont {P.}~\bibnamefont {Jaming}}\ and\ \bibinfo {author} {\bibfnamefont {M.}~\bibnamefont {Matolcsi}},\ }\href {https://doi.org/10.1007/s10474-015-0517-6} {\bibfield  {journal} {\bibinfo  {journal} {Acta Math. Hungarica}\ }\textbf {\bibinfo {volume} {147}},\ \bibinfo {pages} {179} (\bibinfo {year} {2015})}\BibitemShut {NoStop}%
\bibitem [{\citenamefont {Gresch}\ and\ \citenamefont {Kliesch}(2025)}]{gresch23}%
  \BibitemOpen
  \bibfield  {author} {\bibinfo {author} {\bibfnamefont {A.}~\bibnamefont {Gresch}}\ and\ \bibinfo {author} {\bibfnamefont {M.}~\bibnamefont {Kliesch}},\ }\href {https://doi.org/10.1038/s41467-024-54859-x} {\bibfield  {journal} {\bibinfo  {journal} {Nature Communications}\ }\textbf {\bibinfo {volume} {16}},\ \bibinfo {pages} {689} (\bibinfo {year} {2025})}\BibitemShut {NoStop}%
\bibitem [{\citenamefont {Bravyi}\ and\ \citenamefont {König}(2012)}]{Bravyi_2012}%
  \BibitemOpen
  \bibfield  {author} {\bibinfo {author} {\bibfnamefont {S.}~\bibnamefont {Bravyi}}\ and\ \bibinfo {author} {\bibfnamefont {R.}~\bibnamefont {König}},\ }\href {https://link.springer.com/article/10.1007/s00220-012-1606-9} {\bibfield  {journal} {\bibinfo  {journal} {Comm. Math. Phys.}\ }\textbf {\bibinfo {volume} {316}},\ \bibinfo {pages} {641} (\bibinfo {year} {2012})}\BibitemShut {NoStop}%
\bibitem [{\citenamefont {Melo}\ \emph {et~al.}(2013)\citenamefont {Melo}, \citenamefont {Ćwikliński},\ and\ \citenamefont {Terhal}}]{Melo_2013}%
  \BibitemOpen
  \bibfield  {author} {\bibinfo {author} {\bibfnamefont {F.~d.}\ \bibnamefont {Melo}}, \bibinfo {author} {\bibfnamefont {P.}~\bibnamefont {Ćwikliński}},\ and\ \bibinfo {author} {\bibfnamefont {B.~M.}\ \bibnamefont {Terhal}},\ }\href {https://iopscience.iop.org/article/10.1088/1367-2630/15/1/013015} {\bibfield  {journal} {\bibinfo  {journal} {New J. Phys.}\ }\textbf {\bibinfo {volume} {15}},\ \bibinfo {pages} {013015} (\bibinfo {year} {2013})}\BibitemShut {NoStop}%
\bibitem [{\citenamefont {Bosse}(2022)}]{bosse_2022}%
  \BibitemOpen
  \bibfield  {author} {\bibinfo {author} {\bibfnamefont {J.~L.}\ \bibnamefont {Bosse}},\ }\href {https://doi.org/10.5281/zenodo.7303997} {\bibinfo {title} {Floyao}} (\bibinfo {year} {2022})\BibitemShut {NoStop}%
\bibitem [{\citenamefont {González-Cuadra}\ \emph {et~al.}(2023)\citenamefont {González-Cuadra}, \citenamefont {Bluvstein}, \citenamefont {Kalinowski}, \citenamefont {Kaubruegger}, \citenamefont {Maskara}, \citenamefont {Naldesi}, \citenamefont {Zache}, \citenamefont {Kaufman}, \citenamefont {Lukin}, \citenamefont {Pichler}, \citenamefont {Vermersch}, \citenamefont {Ye},\ and\ \citenamefont {Zoller}}]{gonzalez-cuadra_fermionic_2023}%
  \BibitemOpen
  \bibfield  {author} {\bibinfo {author} {\bibfnamefont {D.}~\bibnamefont {González-Cuadra}}, \bibinfo {author} {\bibfnamefont {D.}~\bibnamefont {Bluvstein}}, \bibinfo {author} {\bibfnamefont {M.}~\bibnamefont {Kalinowski}}, \bibinfo {author} {\bibfnamefont {R.}~\bibnamefont {Kaubruegger}}, \bibinfo {author} {\bibfnamefont {N.}~\bibnamefont {Maskara}}, \bibinfo {author} {\bibfnamefont {P.}~\bibnamefont {Naldesi}}, \bibinfo {author} {\bibfnamefont {T.~V.}\ \bibnamefont {Zache}}, \bibinfo {author} {\bibfnamefont {A.~M.}\ \bibnamefont {Kaufman}}, \bibinfo {author} {\bibfnamefont {M.~D.}\ \bibnamefont {Lukin}}, \bibinfo {author} {\bibfnamefont {H.}~\bibnamefont {Pichler}}, \bibinfo {author} {\bibfnamefont {B.}~\bibnamefont {Vermersch}}, \bibinfo {author} {\bibfnamefont {J.}~\bibnamefont {Ye}},\ and\ \bibinfo {author} {\bibfnamefont {P.}~\bibnamefont {Zoller}},\ }\href {https://doi.org/10.1073/pnas.2304294120} {\bibfield  {journal} {\bibinfo  {journal} {Proceedings of the National Academy of Sciences}\ }\textbf
  {\bibinfo {volume} {120}},\ \bibinfo {pages} {e2304294120} (\bibinfo {year} {2023})}\BibitemShut {NoStop}%
\bibitem [{\citenamefont {Seeley}\ \emph {et~al.}(2012)\citenamefont {Seeley}, \citenamefont {Richard},\ and\ \citenamefont {Love}}]{Seeley_2012}%
  \BibitemOpen
  \bibfield  {author} {\bibinfo {author} {\bibfnamefont {J.~T.}\ \bibnamefont {Seeley}}, \bibinfo {author} {\bibfnamefont {M.~J.}\ \bibnamefont {Richard}},\ and\ \bibinfo {author} {\bibfnamefont {P.~J.}\ \bibnamefont {Love}},\ }\href {https://doi.org/10.1063/1.4768229} {\bibfield  {journal} {\bibinfo  {journal} {J. Chem. Phys.}\ }\textbf {\bibinfo {volume} {137}} (\bibinfo {year} {2012})}\BibitemShut {NoStop}%
\bibitem [{\citenamefont {Setia}\ \emph {et~al.}(2019)\citenamefont {Setia}, \citenamefont {Bravyi}, \citenamefont {Mezzacapo},\ and\ \citenamefont {Whitfield}}]{Setia_2019}%
  \BibitemOpen
  \bibfield  {author} {\bibinfo {author} {\bibfnamefont {K.}~\bibnamefont {Setia}}, \bibinfo {author} {\bibfnamefont {S.}~\bibnamefont {Bravyi}}, \bibinfo {author} {\bibfnamefont {A.}~\bibnamefont {Mezzacapo}},\ and\ \bibinfo {author} {\bibfnamefont {J.~D.}\ \bibnamefont {Whitfield}},\ }\href {https://doi.org/10.1103/physrevresearch.1.033033} {\bibfield  {journal} {\bibinfo  {journal} {Phys. Rev. Res.}\ }\textbf {\bibinfo {volume} {1}},\ \bibinfo {pages} {033033} (\bibinfo {year} {2019})}\BibitemShut {NoStop}%
\bibitem [{\citenamefont {Jiang}\ \emph {et~al.}(2018)\citenamefont {Jiang}, \citenamefont {Sung}, \citenamefont {Kechedzhi}, \citenamefont {Smelyanskiy},\ and\ \citenamefont {Boixo}}]{Boxio18}%
  \BibitemOpen
  \bibfield  {author} {\bibinfo {author} {\bibfnamefont {Z.}~\bibnamefont {Jiang}}, \bibinfo {author} {\bibfnamefont {K.~J.}\ \bibnamefont {Sung}}, \bibinfo {author} {\bibfnamefont {K.}~\bibnamefont {Kechedzhi}}, \bibinfo {author} {\bibfnamefont {V.~N.}\ \bibnamefont {Smelyanskiy}},\ and\ \bibinfo {author} {\bibfnamefont {S.}~\bibnamefont {Boixo}},\ }\href {https://journals.aps.org/prapplied/abstract/10.1103/PhysRevApplied.9.044036} {\bibfield  {journal} {\bibinfo  {journal} {Phys. Rev. Appl.}\ }\textbf {\bibinfo {volume} {9}},\ \bibinfo {pages} {044036} (\bibinfo {year} {2018})}\BibitemShut {NoStop}%
\bibitem [{\citenamefont {Oszmaniec}\ \emph {et~al.}(2022)\citenamefont {Oszmaniec}, \citenamefont {Dangniam}, \citenamefont {Morales},\ and\ \citenamefont {Zimbor{\'a}s}}]{oszmaniec22}%
  \BibitemOpen
  \bibfield  {author} {\bibinfo {author} {\bibfnamefont {M.}~\bibnamefont {Oszmaniec}}, \bibinfo {author} {\bibfnamefont {N.}~\bibnamefont {Dangniam}}, \bibinfo {author} {\bibfnamefont {M.~E.}\ \bibnamefont {Morales}},\ and\ \bibinfo {author} {\bibfnamefont {Z.}~\bibnamefont {Zimbor{\'a}s}},\ }\href {https://doi.org/https://doi.org/10.1103/PRXQuantum.3.020328} {\bibfield  {journal} {\bibinfo  {journal} {PRX Quantum}\ }\textbf {\bibinfo {volume} {3}},\ \bibinfo {pages} {020328} (\bibinfo {year} {2022})}\BibitemShut {NoStop}%
\bibitem [{\citenamefont {Nielsen}\ and\ \citenamefont {Chuang}(2000)}]{nielsen00}%
  \BibitemOpen
  \bibfield  {author} {\bibinfo {author} {\bibfnamefont {M.~A.}\ \bibnamefont {Nielsen}}\ and\ \bibinfo {author} {\bibfnamefont {I.~L.}\ \bibnamefont {Chuang}},\ }\href {https://doi.org/https://doi.org/10.1017/CBO9780511976667} {\emph {\bibinfo {title} {Quantum Computation and Quantum Information}}}\ (\bibinfo  {publisher} {Cambridge University Press},\ \bibinfo {year} {2000})\BibitemShut {NoStop}%
\bibitem [{\citenamefont {Cade}\ \emph {et~al.}(2020)\citenamefont {Cade}, \citenamefont {Mineh}, \citenamefont {Montanaro},\ and\ \citenamefont {Stanisic}}]{Montanaro2020}%
  \BibitemOpen
  \bibfield  {author} {\bibinfo {author} {\bibfnamefont {C.}~\bibnamefont {Cade}}, \bibinfo {author} {\bibfnamefont {L.}~\bibnamefont {Mineh}}, \bibinfo {author} {\bibfnamefont {A.}~\bibnamefont {Montanaro}},\ and\ \bibinfo {author} {\bibfnamefont {S.}~\bibnamefont {Stanisic}},\ }\href {https://journals.aps.org/prb/abstract/10.1103/PhysRevB.102.235122} {\bibfield  {journal} {\bibinfo  {journal} {Phys. Rev. B}\ }\textbf {\bibinfo {volume} {102}},\ \bibinfo {pages} {235122} (\bibinfo {year} {2020})}\BibitemShut {NoStop}%
\bibitem [{\citenamefont {Banica}\ \emph {et~al.}(2012)\citenamefont {Banica}, \citenamefont {Nechita},\ and\ \citenamefont {{\.Z}yczkowski}}]{banica12}%
  \BibitemOpen
  \bibfield  {author} {\bibinfo {author} {\bibfnamefont {T.}~\bibnamefont {Banica}}, \bibinfo {author} {\bibfnamefont {I.}~\bibnamefont {Nechita}},\ and\ \bibinfo {author} {\bibfnamefont {K.}~\bibnamefont {{\.Z}yczkowski}},\ }\href {https://doi.org/https://doi.org/10.1142/S1230161212500242} {\bibfield  {journal} {\bibinfo  {journal} {Open Syst. Inf. Dyn.}\ }\textbf {\bibinfo {volume} {19}},\ \bibinfo {pages} {1250024} (\bibinfo {year} {2012})}\BibitemShut {NoStop}%
\bibitem [{\citenamefont {McClean}\ \emph {et~al.}(2020)\citenamefont {McClean} \emph {et~al.}}]{McClean_2020}%
  \BibitemOpen
  \bibfield  {author} {\bibinfo {author} {\bibfnamefont {J.~R.}\ \bibnamefont {McClean}} \emph {et~al.},\ }\href {https://doi.org/10.1088/2058-9565/ab8ebc} {\bibfield  {journal} {\bibinfo  {journal} {Quantum Science and Technology}\ }\textbf {\bibinfo {volume} {5}},\ \bibinfo {pages} {034014} (\bibinfo {year} {2020})}\BibitemShut {NoStop}%
\bibitem [{\citenamefont {Javadi-Abhari}\ \emph {et~al.}(2024)\citenamefont {Javadi-Abhari}, \citenamefont {Treinish}, \citenamefont {Krsulich}, \citenamefont {Wood}, \citenamefont {Lishman}, \citenamefont {Gacon}, \citenamefont {Martiel}, \citenamefont {Nation}, \citenamefont {Bishop}, \citenamefont {Cross}, \citenamefont {Johnson},\ and\ \citenamefont {Gambetta}}]{qiskit2024}%
  \BibitemOpen
  \bibfield  {author} {\bibinfo {author} {\bibfnamefont {A.}~\bibnamefont {Javadi-Abhari}}, \bibinfo {author} {\bibfnamefont {M.}~\bibnamefont {Treinish}}, \bibinfo {author} {\bibfnamefont {K.}~\bibnamefont {Krsulich}}, \bibinfo {author} {\bibfnamefont {C.~J.}\ \bibnamefont {Wood}}, \bibinfo {author} {\bibfnamefont {J.}~\bibnamefont {Lishman}}, \bibinfo {author} {\bibfnamefont {J.}~\bibnamefont {Gacon}}, \bibinfo {author} {\bibfnamefont {S.}~\bibnamefont {Martiel}}, \bibinfo {author} {\bibfnamefont {P.~D.}\ \bibnamefont {Nation}}, \bibinfo {author} {\bibfnamefont {L.~S.}\ \bibnamefont {Bishop}}, \bibinfo {author} {\bibfnamefont {A.~W.}\ \bibnamefont {Cross}}, \bibinfo {author} {\bibfnamefont {B.~R.}\ \bibnamefont {Johnson}},\ and\ \bibinfo {author} {\bibfnamefont {J.~M.}\ \bibnamefont {Gambetta}},\ }\href {https://doi.org/10.48550/arXiv.2405.08810} {\bibinfo {title} {Quantum computing with {Q}iskit}} (\bibinfo {year} {2024}),\ \Eprint {https://arxiv.org/abs/arXiv:2405.08810} {arXiv:2405.08810} \BibitemShut
  {NoStop}%
\end{thebibliography}%

\renewcommand{\appendixname}{Supplementary Information}

\appendix

\onecolumngrid
\section{Number of measurement settings ensuring all quadruples are measured in 2D setting}\label{app:quadropNumbr}

Below we will show that at most seven ($7$) vertical permutations of the rectangular layout of modes considered in the main text suffice to ensure measurability of all Majorana quadruples (for the case when $L+1$  is a prime number).
Let $\mathcal{P}$ be an arrangement of $\tilde{N}=l\cdot k$ different symbols placed in square boxes arranged in an $k\times l$ rectangular pattern (note that in what follows $l$ plays the role of $L+1$ from the main text). Let 
\begin{equation}\label{eq:permVertical}
    \pi=\pi^{(1)}_{v_1} \pi_{v_2}^{(2)}\ldots \pi^{(k)}_{v_k}
\end{equation}
be a permutation consisting of $k$ independent cyclic shifts $\pi^{(i)}_{v_i}$, acting on the columns $i=1,\ldots,k$ of $\mathcal{P}$, by distance $v_i\in\{0,1,\ldots,l-1\}$ (the shifts are modulo $l$, because there are only this many elements in each column of $\mathcal{P}$). Let $\pi(\mathcal{P})$ denote the transformed version of $\mathcal{P}$ in which every box of $\mathcal{P}$ is moved according to the action of the permutation $\pi$.

\begin{lem}
    Let $l\geq 7$ be a prime number and let $k\geq 4$. Consider a permutation $\pi$ as in \eqref{eq:permVertical} chosen such that $v_i\neq v_j$ for $i\neq j$ ($i,j\in [k]$). Then for every quadruple $(p_1,p_2,p_3,p_4)$ of distinct elements in $\mathcal{P}$, the elements $p_i$ can be found in different rows of at least one of the seven arrangements $\mathcal{P},\pi(\mathcal{P}),\ldots,\pi^6(\mathcal{P})$, where $\pi^a$ denotes the $a$--th iteration of the permutation $\pi$.
    
\end{lem}

\begin{proof}
    To every quadruple $\bm{p}=(p_1,p_2,p_3,p_4)$ we can assign four vertical positions $y_1,y_2,y_3,y_4\in [l]$  of the elements $p_i$ in the array. We will assume that elements of $\bm{p}$ belong to four distinct columns of $\mathcal{P}$ and furthermore, without loss of generality, that element $p_i$ belongs to the $i$--th column (the proof for the case when some $p$'s belong to the same column is analogous). After $r=0,1,\ldots,$ applications of  $\pi$ to $\bm{p}$, the vertical positions of $y_i$ take the form
    \begin{equation}
        y_i (r)= y_i(0) +r\cdot v_i \ \mod l\ . 
    \end{equation}
    The evolution of the differences between the vertical positions $\Delta_{ij}=y_i-y_j$ is analogous:
    \begin{equation}\label{eq:diferencesEVOLUTION}
    \Delta_{ij}(r)= \Delta_{ij}(0) + r\cdot v_{ij} \ \mod l\ ,
    \end{equation}
where $v_{ij}=v_i-v_j \neq 0$. Note that elements $\pi^r(\bm{p})$ are in distinct rows of $\pi^r(\mathcal{P})$ if and only if $\Delta_{ij}(r) \neq 0$ for all $i>j$. Therefore we need to ensure that for at least one $r=0,1,\ldots,6$ we have $\Delta_{ij}(r)\neq 0$ for all $i>j$. This follows by trying to find the ``time'' $r_{ij}$ that collisions between elements $p_i$ and $p_j$ take place: $\Delta_{ij}(r_{ij})=0$. Since this is a linear equation for $r_{ij}$, and $l$ was set to be a prime number, it has a unique solution $r_{ij}^\ast \in\{0,1,\ldots,l-1\}$. Note that since we are looking at four elements in $\bm{p}$, there are only $\binom{4}{2}=6$ equations $\Delta_{ij}(r_{ij})=0$ to consider. Each one has a unique solution $r_{ij}^\ast$ and there are at most 6 of them. Therefore, for at least one $r\in\{0,1,\ldots,6\}$  we have $\Delta_{ij}(r) \neq 0$ for all $i>j$.
\end{proof}

\section{Four measurement rounds are sufficient to estimate Hamiltonians}\label{app:4rounds}

We now show that the measurement scheme defined in Sec. \ref{sec:physical_hamiltonians} (see also Fig. \ref{fig:rounds}) guarantees that every $\gamma_A$ in the decomposition of $H$ (i.e. $A\in\mathcal{X}_2\cup\mathcal{X}_4$) is measured with visibility $\eta_A=\Theta(N^{-|A|/4})$. Following similar arguments made in Prop. \ref{prop:Randomized}, the observable $\gamma_A$ is jointly measured in the $r$--th measurement round (with the required visibility) if $A\in\mathcal{M}^{(r)}$, i.e. the elements of $A$ lie in distinct clusters associated with the unitary $U^{(r)}$. Therefore, it is sufficient to prove the following:
\begin{lem}
For each $A\in\mathcal{X}_2\cup\mathcal{X}_4$, $\exists\, r\in \{1,2,3,4\}$ such that $A\in\mathcal{M}^{(r)}$.
\end{lem}
\begin{proof}
The relevant quadratic observables $\gamma_A$ contain both an even and odd Majorana term. In every round, the even (odd) term belongs in an even (odd) cluster, therefore $A\in\mathcal{M}^{(r)}$ for all $r=1,\ldots,4$. The relevant quartic observables $\gamma_A$ contain two even and two odd terms, therefore $A=\{i,j,k,l\}\in\mathcal{M}^{(r)}$ if $i,j,k,l$ lie in two even and two odd clusters. In the first round, $A\in\mathcal{M}^{(1)}$ if the set $A$ has no even or odd pairs in the same cluster (e.g. $\gamma_{1,2,7,8}$ from Fig. \ref{fig:rounds}, left). If $A\notin\mathcal{M}^{(1)}$ there are three cases to consider: (a) the odd pair in $A$ (but not the even pair) is in one cluster (e.g. $\gamma_{1,2,3,8}$); (b) the even pair (but not the odd pair) is in one cluster (e.g. $\gamma_{1,2,4,7}$); (c) the odd and even pairs are in the same odd and even clusters (e.g. $\gamma_{1,2,3,4}$). It is easy to check---due to the vertical shifts applied in each round (see Fig. \ref{fig:rounds})---that when case (a) applies, $A\in\mathcal{M}^{(2)}$; when case (b) applies, $A\in\mathcal{M}^{(3)}$; and when case (c) applies, $A\in\mathcal{M}^{(4)}$. Hence, every quadruple $A\in\mathcal{X}_4$ is an element of at least one $\mathcal{M}^{(r)}$.
\end{proof}

\section{Details of the strategy for physical fermionic Hamiltonians}
\label{sec:details_strategy_physical_ham}
Here we discuss the details of the postprocessing in the strategy targeting physical Hamiltonians, described in Sec. \ref{sec:physical_hamiltonians}. Specifically we explain how the outcome $\e^{(r)}_A$ needed to define the estimator $\hat{\gamma}^{(r)}_A$ in Eq. \eqref{eq:estimator_av} is obtained from the measurement outcome in step (iii) of the general strategy.

Recall that the outcome $e^{(r)}_A=s_{AB}(-1)^{|A\cap X|} \prod_{j\in B'} q_i$ and visibility $\eta^{(r)}_A=|\det(R^{(r)}_{A,B})|$ implicitly depend on $B'\subset[N]$, which labels the subset of Majorana pairs $B$ used in step (iii) of the general protocol. For each $A$, we choose a unique $B'$ (and thus also $B$) to maximize $\eta^{(r)}_{A}$. For pairs there is one obvious choice: if $A=\{i,j\}$ with $i\in E^{(r)}_\alpha$, $j\in O^{(r)}_\beta$, and $(p,q)$ is the coupling between the two clusters, then $B=\pi^{-1}\circ (\sigma^{(r)})^{-1}(p,q)$. For quadruples there is an extra freedom: if $A=\{i,j,k,l\}$ has elements in the clusters $E_\alpha, E_{\alpha'},O_{\beta},O_{\beta'}$, we can choose the couplings of $(E^{(r)}_{\alpha},O^{(r)}_{\beta})$ and $(E^{(r)}_{\alpha'},O^{(r)}_{\beta'})$ or, alternatively, $(E^{(r)}_{\alpha},O^{(r)}_{\beta'})$ and $(E^{(r)}_{\alpha'},O^{(r)}_{\beta})$. In this scenario, we select the two couplings $(p,q,r,s)$ that maximize $\eta^{(r)}_A$ and set $B=\pi^{-1}\circ(\sigma^{(r)})^{-1}(p,q,r,s)$. 

\section{Variance of the Hamiltonian estimator}\label{app:variance}
We can write the estimator defined in \eqref{eq:hamiltonian_estimator} as $\hat H=\sum_r\hat H^{(r)}$, where $\hat H^{(r)}=\sum_A\alpha_A^{(r)}h_A\hat\gamma_A^{(r)}$ involves estimators of the relevant Majorana observables from the $r$--th round. Each round is independent, therefore 
\begin{align}
\label{eq:variance_estimator_physical}
\mathrm{Var}[\hat H]&=\sum_{r=1}^4\left(\ev{(\hat H^{(r)})^2}-\left(\ev{\hat H^{(r)}}\right)^2\right)\\
&=\sum_{r=1}^4 \left(\sum _{A,{A'} \in \mathcal{X}_2\cup\mathcal{X}_4} \alpha_A^{(r)} \alpha_{A'}^{(r)} h_A h_{A'} \left(\ev{\hat{\gamma}_A^{(r)} \hat{\gamma}_{A'}^{(r)}} - \ev{\hat{\gamma}_A^{(r)}}\ev{\hat{\gamma}_{A'}^{(r)}}\right)\right)\\
&=\sum_{r=1}^4 \left(\sum _{A,{A'} \in \mathcal{M}^{(r)}} \alpha_A^{(r)} \alpha_{A'}^{(r)} h_A h_{A'} \left(\ev{\hat{\gamma}_A^{(r)} \hat{\gamma}_{A'}^{(r)}} - \tr(\gamma_A\rho)\tr(\gamma_{A'}\rho)\right)\right)\,,\label{eq:molecular_variance}
\end{align}
where $\mathcal{M}^{(r)}\subset\mathcal{X}_2\cup\mathcal{X}_4$ is the subset of Majorana terms consistent with $U^{(r)}$. Following the calculation presented in Appendix I.2 in \cite{MCO2024} we have,
\begin{equation}\label{eq:covariance_gamma_A_gamma_B}
 \ev{\hat{\gamma}_A^{(r)}\hat{\gamma}_{A'}^{(r)}} =\begin{cases}
\delta_{|A \vartriangle A'|,|f(A)\vartriangle f(A')|}\frac{\nu^{(r)}_{A\vartriangle A'}}{\nu^{(r)}_A\nu^{(r)}_{A'}}\tr(\rho\,\gamma_{A\vartriangle A'}) &\mbox{if} \,\,\, A\neq A' \\
(1/\nu^{(r)}_A)^2  &\mbox{otherwise} \,,
 \end{cases}
\end{equation}
where $\nu^{(r)}_{A}=\det(R^{(r)}_{A,f(A)})$, $\nu^{(r)}_{A'}=\det(R^{(r)}_{A',f(A')})$ and $\nu^{(r)}_{A\triangle A'}=\det(R^{(r)}_{A\vartriangle A',f(A)\vartriangle f(A')})$. The matrix $R^{(r)} \in O(2N)$ corresponds to the FLO unitary $U^{(r)}$ implemented in the $r$--th round and $A\vartriangle A'$ denotes the symmetric difference of $A$ and $A'$, i.e., $A\vartriangle A'=(A\cup A')\setminus (A\cap A')$. We have introduced the notation $f(A)\subset[2N]$ to label the subset of Majorana pairs that are measured in step (iii) of the general joint measurement protocol. For each $A$, the set $f(A)$ implicitly depends on $r$, and is chosen uniquely (for a given $r$) to optimize $\nu_A^{(r)}$, as described in Sec. \ref{sec:physical_hamiltonians}.

\section{Optimization of coefficients in estimator for physical fermionic Hamiltonians}
\label{app:optimization}

For any given Hamiltonian and state, we can optimize the coefficients $\alpha^{(r)}_A$ which appear in the estimator \eqref{eq:hamiltonian_estimator} of the Hamiltonian in order to minimize the variance \eqref{eq:variance_estimator_physical}. Let $h$ be the number of terms in the Hamiltonian \eqref{eq:chemistryHamiltonian} with non-zero coefficients $h_A$. For each round $r$, let $\boldsymbol{\alpha}^{(r)} = (\alpha^{(r)}_{A_1}, \alpha^{(r)}_{A_2}, ..., \alpha^{(r)}_{A_h})$ be a vector of dimension $h$ with indices $A_i \in \mathcal{X}_2 \cup \mathcal{X}_4 $ such that $ h_{A_i} \neq 0$, and let $\mathbf{C}^{(r)}$ denote the $h \times h$ matrix with entries $C^{(r)}_{A_i, A_j} = h_{A_i} h_{A_j} \,\, \mathrm{Cov}(\hat{\gamma}_{A_i}^{(r)}\hat{\gamma}_{A_j}^{(r)})$, where $\mathrm{Cov}(\hat{\gamma}_{A_i}^{(r)}\hat{\gamma}_{A_j}^{(r)}) = \ev{\hat{\gamma}_{A_i}^{(r)}\hat{\gamma}_{A_j}^{(r)}} - \ev{\hat{\gamma}_{A_i}^{(r)}}\ev{\hat{\gamma}_{A_j}^{(r)}}$. We assume the same ordering of indexing terms $A_1, ... , A_h$ for $\boldsymbol{\alpha}^{(r)}$ and $\mathbf{C}^{(r)}$. The variance \eqref{eq:variance_estimator_physical} can then be expressed as
\begin{equation}
\label{eqn:variance_matrix_form}
    \mathrm{Var}[\hat{H}] = \sum_{r=1}^4 \left(\boldsymbol{\alpha}^{(r)}\right) ^{\mathrm{T}} \mathbf{C}^{(r)} \boldsymbol{\alpha}^{(r)}\,.
\end{equation}
We now express the optimal coefficients $\boldsymbol{\alpha}^{(r)}_{\mathrm{opt}}$ in terms of the matrix $\mathbf{D}^{(r)} = \frac{1}{2}\left( \mathbf{C}^{(r)} \right) ^{-1}$, where the inverse is only taken on the submatrix $\mathbf{C}^{(r)}|_{\mathcal{M}^{(r)}}$ (which we assume to be invertible) and the remaining rows and columns are left to be zero.
\begin{lem}
The coefficients in the estimator \eqref{eq:hamiltonian_estimator} which minimize the variance \eqref{eq:variance_estimator_physical} are given by
\begin{equation}
\label{eq:optimized_alphas}
    \boldsymbol{\alpha}^{(r)}_{\mathrm{opt}} = \mathbf{D}^{(r)} \left( 
\sum_{r=1}^4 \mathbf{D}^{(r)} \right)^{-1} \mathbf{1},
\end{equation}
where $\mathbf{1}$ is the vector of ones, of dimension $h$.
\end{lem}
\begin{proof}
This can be seen by the constrained minimization of \eqref{eqn:variance_matrix_form} under the constraint $\sum_{r=1}^4 \alpha_A^{(r)} = 1$ for each $A$, or equivalently, $\sum_{r=1}^4 \boldsymbol{\alpha}^{(r)} = \mathbf{1}$.
The problem can be solved by minimizing 
\begin{equation}
\label{eqn:lagrangian}
    \mathcal{L} = \sum_{r=1}^4 \left(\boldsymbol{\alpha}^{(r)}\right) ^{\mathrm{T}} \mathbf{C}^{(r)} \boldsymbol{\alpha}^{(r)} - \sum_A \mu_A \left( \sum_{r=1}^4 \alpha_A^{(r)} - 1\right),
\end{equation}
where $\mu_A$ are Lagrange multipliers.
At the minimum, we require
\begin{equation}
    \label{eq:lagrangian_derivative}
     \frac{\partial \mathcal{L}}{\partial \alpha_A^{(r)}} = 2 \sum_{A'} \alpha_{A'}^{(r)} 
C^{(r)}_{A, A'} - \mu_A=0\,,
\end{equation}
for every $A$ and $r$. Solving the above equation for $\alpha^{(r)}_A$ yields 
\begin{equation}
\label{eqn:lagrange_alpha_of_mu}
    \boldsymbol{\alpha}^{(r)} = \mathbf{D}^{(r)} \, \boldsymbol{\mu}\,,
\end{equation}
where $\boldsymbol{\mu}$ denotes the vector whose components are the $h$ Lagrange multipliers, with ordering consistent with $\boldsymbol{\alpha}^{(r)}$. Substituting this into the original constraint we obtain $\sum_r \mathbf{D}^{(r)} \boldsymbol{\mu} = \mathbf{1}$ and therefore $\boldsymbol{\mu} = (\sum_r \mathbf{D}^{(r)})^{-1} \mathbf{1}$. A further substitution into \eqref{eqn:lagrange_alpha_of_mu} gives \eqref{eq:optimized_alphas}. 

Note that \eqref{eq:lagrangian_derivative} correctly defines the minimum of the variance \eqref{eqn:variance_matrix_form}, since  $\mathbf{C}^{(r)}=\mathrm{ diag(\mathbf{h})}^T\cdot \mathbf{Cov}^{(r)} \cdot\mathrm{diag(\mathbf{h})}$ is positive-semidefinite. Here, $\mathrm{diag}(\mathbf{h})$ is the diagonal matrix with entries $h_A$, and the covariance matrix $\mathbf{Cov}^{(r)}$, with entries $\mathrm{Cov}(\hat{\gamma}_{A_i}^{(r)}\hat{\gamma}_{A_j}^{(r)})$, is always positive-semidefinite.
 Furthermore, the value of $\boldsymbol{\alpha}^{(r)}$ at which \eqref{eq:lagrangian_derivative} is satisfied is the only viable candidate for the minimum of \eqref{eqn:variance_matrix_form} under the given constraint. In general, a minimum at $\boldsymbol{\alpha}^{(r)}$ could also be found when the gradient of the constraint function is zero (which in our problem is never the case) or when the gradient of the minimized function itself \eqref{eqn:variance_matrix_form} is zero (which in our case is $\alpha_A^{(r)} = 0$ for any $A,r$, but this does not satisfy the constraint).
\end{proof}

\begin{table}[h]
\begin{tabular}{|c|c|c|c|c|c|c|}
\hline
\begin{tabular}[c]{@{}c@{}}Molecule\\ (qubits)\end{tabular} &  \begin{tabular}[c]{@{}c@{}}Uniform\\ coefficients \end{tabular}  &  \begin{tabular}[c]{@{}c@{}}Optimized\\ coefficients \end{tabular} \\ \hline
H2 (8)                                                      &  341   & 26.6                                             \\ \hline
H2 (8), almost Hadamard matrix                                                      &  63.3   & 49.5                                            \\ \hline
LiH (12)                                                    &  128   & 106                                              \\ \hline
BeH2 (14)                                                   &  281   & 245                                              \\ \hline
H2O (14)                                                    &  1620 & 1492                                               \\ \hline
NH3 (16)                                                    & 1157 & 1061                                              \\ \hline
\end{tabular}
\caption{\textbf{Comparison of variances for estimating energies of standard benchmark molecules using the joint measurability strategy.} \\ The coefficients $\alpha^{(r)}_A$ in Eq. \eqref{eq:hamiltonian_estimator} are either uniform \eqref{eqn:uniform_coefficients} or optimized \eqref{eq:optimized_alphas} for the ground state of the Hamiltonian. The variances are calculated for the ground states of five molecular Hamiltonians, with units Ha$^2$ (see also Table \ref{tab:benchmark} of the main text which summarizes only the optimized coefficients). Except for H2, in which both lower-flat \cite{Jaming15} and almost Hadamard \cite{banica12} matrices are used, we typically apply the lower-flat orthogonal matrices defined in Appendix B.1. of \cite{MCO2024}.}
\label{tab:benchmark_appendix}
\end{table}

Table \ref{tab:benchmark} of the main text summarizes the variance of the estimator, with $\alpha^{(r)}_A$ optimized via \eqref{eq:optimized_alphas} for the ground state of the Hamiltonian, obtained numerically. However, since this optimization depends on the state, it may not always be practically applicable, e.g. if one has no information about the state to be measured. Moreover, for more complex molecules, the optimization appears not to provide a significant advantage over uniform coefficients. Uniform coefficients have the form
\begin{equation}
\label{eqn:uniform_coefficients}
    \alpha_A^{(r)}= \begin{cases} 1/f \, &\mbox{if} \, A\in \mathcal{M}^{(r)}\\
    0 \, &\mbox{otherwise},
    \end{cases}
\end{equation} where $f$ is the number of times $A$ appears in $\{\mathcal{M}^{(r)}\}_{r=1}^4$ and $\mathcal{M}^{(r)}$, first defined in Sec. \ref{sec:constructionOFunitaries}, is the set of terms $A$ which are estimated in $r$-th round.

In Table \ref{tab:benchmark_appendix} we present a comparison of the variances for benchmark Hamiltonians using both optimized and uniform coefficients. For H2 we test both a lower-flat orthogonal matrix \cite{Jaming15} (see also \cite{MCO2024}) and an almost Hadamard matrix \cite{banica12}, with $-\frac{1}{3}$ on the diagonal and $\frac{2}{3}$ elsewhere. The choice of orthogonal matrix, and its corresponding unitary $U^{(r)}$, significantly affects the variance of the estimator.

\section{Details of the numerical benchmarks on molecular Hamiltonians}
\label{app:numerics}
The variance of our estimator for molecular Hamiltonians, as presented in Table \ref{tab:benchmark}, is calculated in the following way. We obtain the Hamiltonian, represented as a linear combination of Majorana operators, from OpenFermion \cite{McClean_2020} interfaced with the PySCF quantum chemistry package. The orbital bases used are the 6-31G basis for the H2 molecule and the STO-3G basis for the remaining molecules, as is standard in the literature (see e.g. \cite{Zhao21}). To obtain the ground state we represent the Hamiltonian in the Pauli basis via the Jordan-Wigner encoding. For efficiency we use Qiskit's \cite{qiskit2024} sparse matrix representation of the Hamiltonian and numerically find the vector corresponding to the smallest eigenvalue using scipy's function "eigsh" with the default solver. To calculate the variance \eqref{eq:molecular_variance}, it is necessary to evaluate expectation values of Majorana monomials (see Eq. \eqref{eq:covariance_gamma_A_gamma_B}) on the ground state. For this purpose we employ Qiskit's efficient calculation of expectation values on statevectors. 
We consider Hamiltonians for molecules with atoms around equilibrium distances. The coordinates (in units \AA) of the atoms are as follows:

H2: H: 0.0, 0.0, 0.0; H: 0.0, 0.0, 0.7414,

LiH: Li: 0.0, 0.0, 0.0; H: 0.0, 0.0, 1.5949,

BeH2: Be: 0.0, 0.0, 0.0; H: 0.0, 0.0, 1.3264; H: 0.0, 0.0, -1.3264,

H2O: H: 0.0, 0.7572, -0.4692; H: 0.0, -0.7572, -0.4692; O: 0.0, 0.0, 0.1173,

NH3: N: 0.0, 0.0, 0.0; H: 0.0, -0.9377, -0.3816; H: 0.8121, 0.4689, -0.3816; H: -0.8121, 0.4689, -0.3816.

The orthogonal lower-flat matrices $R^{(\alpha)}$ used to construct the orthogonal matrix \eqref{eq:SuperposingOrthogonal}, which realizes the superposition of modes within their clusters, are constructed according to the recipe in Appendix B.1. of \cite{MCO2024}. For a given system size we use the same matrix for any $\alpha$.

\rev{
\section{Details on the numerical simulation of the scheme for fermionic Gaussian states}
\label{app:gaussian_states_simulation}

\begin{algorithm}[]
\label{alg:sampling}
 \KwData{number of samples $s$;
 
 covariance matrix $C_\rho$ of Gaussian state $\rho$;
 
 orthogonal matrices $R_U^{(r)}$ ($r \in [4]$) defining FLO transformations $U^{(r)}$ of the rounds from Section \ref{sec:physical_hamiltonians}.} 
 
 \KwResult{$s \times 4$ samples $\{ 
\bm{q}_j^{(r)} \} \, (j \in [s], r \in [4]$), $\bm{q}=(q_1,q_2,\ldots, q_N)$, $q_i = \pm 1$ \Comment*{(step (iii) from Section \ref{sec:Genprotocol})}

 bit-strings $\{X^{(r)}_j\} \, (j \in [s], r \in [4])$ corresponding to random Majoranas drawn for each sample.
}
 
\For{$j\gets1$ \KwTo $s$}
{
\For{$r \gets 1$ \KwTo $4$}
{
  $C \gets C_\rho$ \Comment*{initialize a covariance matrix of the input state}
  
  draw random subset $X \subseteq 2N$ defining Majorana $\gamma_X$\;

  $X^{(r)}_j \gets X$\;

  construct $R_X$  with matrix elements $R_{X, kl} = -1 ^ {|X|+|X \cap \{k\}|}\delta_{kl}$, s.t. $\forall_{k \in [2N]}$ $\gamma_X \gamma_k \gamma_X^\dagger = \sum_l R_{X, kl} \gamma_l$; 

  $C \gets R_U^{(r)T} R_X^T C R_X R_U^{(r)}$ \Comment*{ evolve the state with $\gamma_X^\dagger$ and $U^\dagger$}

  $\bm x \gets\mathrm{FLOYao.sample}(C)$ ($\bm x = (x_1, x_2, ...., x_N), x\in \{0, 1\}$)  \Comment*[r]{measure in computational basis}
  
  $q_i \gets 1 - 2 x_i$, where  $\bm{q}_j^{(r)} = (q_1, q_2, ..., q_N); $ 
 }
 }
 \caption{Classical simulation of the sampling procedure for estimating physical Hamiltonians (Section \ref{sec:physical_hamiltonians})}
\end{algorithm}

\newpage

\begin{algorithm}[]
\label{alg:postprocessing}
\KwData{$s \times 4$ samples $\{ 
\bm{q}_j^{(r)} \} \, (j \in [s], r \in [4]$), $\bm{q}=(q_1,q_2,\ldots, q_N)$, $q_i = \pm 1$;

orthogonal matrices $R_U^{(r)}$ ($r \in [4]$) defining FLO transformations $U^{(r)}$ of the rounds from Section \ref{sec:physical_hamiltonians};

bit-strings $\{X^{(r)}_j\} \, (j \in [s], r \in [4])$ corresponding to random Majoranas drawn for each sample;

Hamiltonian $H=\sum_{A\in \mathcal{X}_{2}\cup \mathcal{X}_{4}} h_A \gamma_A\ $;

coefficients $\boldsymbol{\alpha}^{(r)} = (\alpha^{(r)}_{A_1}, \alpha^{(r)}_{A_2}, ..., \alpha^{(r)}_{A_h})$ \Comment*[r]{uniform or optimized, see Supplementary Information \ref{app:optimization}}
}

\KwResult{Empirical estimator of Hamiltonian $\hat{\bar{H}}$ and variance $\mathrm{Var}(\hat{\bar{H}})$.}

\For{$j\gets1$ \KwTo $s$}
{
\For{$r \gets 1$ \KwTo $4$}
{

\For{$A \in \mathcal X_2 \cup \mathcal X_4$}
{

\If{$A \in \mathcal{M}^{(r)}$}
{
choose $B' \subset [N], B \subset [2N]$ according to Supplementary Information \ref{sec:details_strategy_physical_ham}; 

$\eta_{A}^{(r)} \gets {\mathrm{det}(R_{U, AB}^{(r)})}$ \Comment*{sharpness}

$e_{A} \gets \mathrm{sgn} (\mathrm{det}(R_{U, AB})) (-1)^{|A \cap X|} \prod_{k\in B'} q^{(r)}_{j,k} $ \Comment*{step (iv) from Section \ref{sec:Genprotocol}}

 $\hat{\gamma}_A^{(r)} \gets \frac{1}{\eta_{A}^{(r)}} e_{A} $;  
} 
 \Else{$\hat{\gamma}_A^{(r)} \gets 0$ \Comment*{$\gamma_A$ not estimated from this round}} 
}
$H^{(r)} \gets \sum_A\alpha_A^{(r)}h_A\hat\gamma_A^{(r)}$;
}
$H_j \gets \sum_{r=1} ^4 H^{(r)}$;
}

$\hat{\bar{H}} \gets \frac{1}{s} \sum_{j=1}^s H_j$;

$\mathrm{Var}(\hat{\bar{H}}) \gets \frac{1}{s-1} \sum_{j=1}^s (H_j - \hat{\bar{H}})^2$;
 
\caption{Postprocessing of samples from Algorithm \ref{alg:sampling} for estimation of Hamiltonian}
\end{algorithm}
}

We simulate the strategy for physical Hamiltonians from Section \ref{sec:physical_hamiltonians} for random fermionic Gaussian states, which is classically efficient \cite{terhal02, Bravyi_2012}, and we numerically validate the accuracy of the variance formula \eqref{eq:molecular_variance} for Hamiltonian estimation. The simulation yields energy estimates that agree with analytical values, up to statistical precision. For each molecular Hamiltonian we draw a different random Gaussian state and simulate our joint measurement strategy using $\mathcal{O}(10^5)$ rounds. The energies of each state are of the order of 1, and the values from analytical calculations and numerical estimations are in agreement up to at least one decimal place. This is orders of magnitude smaller than the variances, which are of orders comparable to the ground states in Table \ref{tab:benchmark} ($10-10^3$).

    Above in Algorithm \ref{alg:sampling} we present the procedure of obtaining the samples for the measurement scheme from classical simulation and in Algorithm \ref{alg:postprocessing} the postprocessing of the samples for Hamiltonian estimation. For the classical simulation we use the package FLOYao.jl \cite{bosse_2022} (see documentation of the package and \cite{Melo_2013} for an accessible introduction to the topic).

\section{Shallow circuits for \emph{vertical} FLO unitaries in 2D layouts under the JW encoding}\label{app:decomp}

To implement our joint measurement on a 2D grid of qubits we will use a horizontal JW encoding of the Majorana modes (described below). The non-trivial aspect of our implementation is the FLO transformation $U^{(r)}$ which consists of the product of three unitaries $U_{\mathrm{sup}}$, $U_{\mathrm{pair}}$ and $U^{(r)}_{\mathrm{resh}}$, introduced in Sec. \ref{sec:constructionOFunitaries}.
As we shall see, our choice of horizontal JW encoding ensures it is straightforward to implement horizontal transformations such as $U_{\mathrm{sup}}$ in gate depth $\mathcal{O}(N^{1/2})$ with $\mathcal{O}(N^{3/2})$ nearest-neighbor two-qubit gates. On the other hand,  implementing vertical transformations such as $U_{\mathrm{pair}}$ and $U^{(r)}_{\mathrm{resh}}$ in a horizontal encoding (with the same gate depth and counts) is less straightforward.

In this section we will show that a class of unitaries we call JW-vertical FLO transformations, which include $U_{\mathrm{pair}}$ and $U^{(r)}_{\mathrm{resh}}$, can be implemented in the horizontal JW encoding with the required gate depth and gate count. The basic idea, as demonstrated in the proof of Prop. \ref{prop:compilation}, is to find a suitable map between the relevant circuit in the vertical JW encoding (for which the implementation is trivial) and its realization in the horizontal JW encoding.

\subsection{Majorana mode labellings}
\begin{figure*}[h!]
\centering
\includegraphics[width=9cm]
{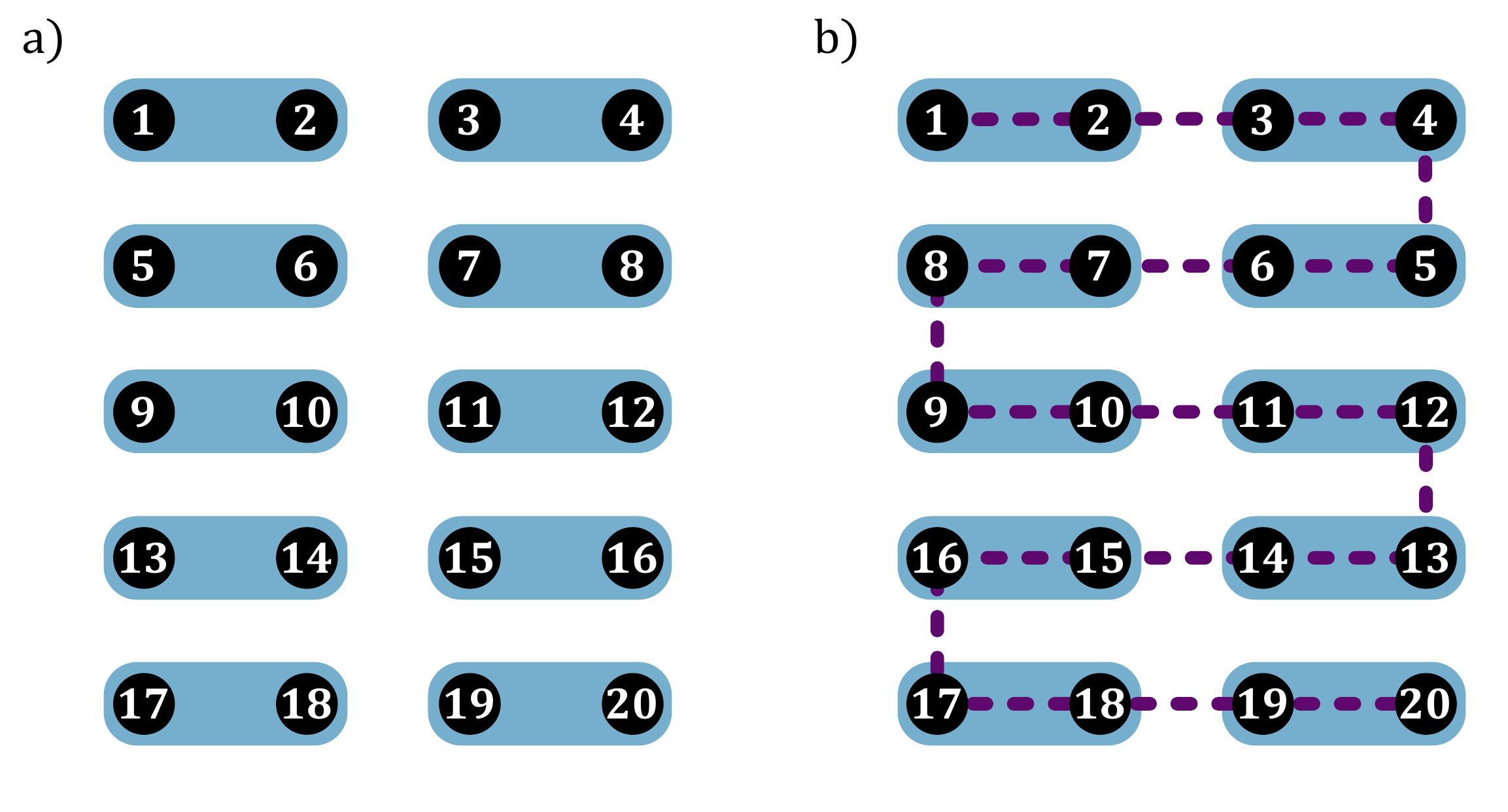}
	\caption {\textbf{Two different ways of labelling $2N= 2l_x l_y$ Majorana modes in a $l_x\times l_y$ rectangular qubit lattice for $l_x=2,l_y=5$.} \\ Black dots represent Majorana modes while qubits are represented by blue rounded rectangles. Part a) is the standard labelling (left to right in each row) used througout the main text. Part b) is a snake-like labelling used in \cite{Boxio18,Montanaro2020} and in this section.}
   \label{fig:2Dlabeling}
\end{figure*}

In what follows we will consider a \textit{snake} labelling of Majorana modes (see Fig. \ref{fig:2Dlabeling}b) in a 2D layout. In particular, modes are arranged along a "snake" which traverses subsequent rows from left to right and then right to left. This differs slightly from the labelling used in the main text (see Fig. \ref{fig:2Dlabeling}a), where the modes were arranged from left to right in each row.

\textbf{Remark.} However, both labellings can be related by a \emph{horizontal} transformation of the modes $P=\oplus_{\alpha -\text{even}} O_\alpha$, which reverses the order of Majorana modes in even rows labelled by $\alpha$ (for the example with $l_x=2,l_y=5$, we have that $O_2$ realizes $\sigma_2=(5,8)(6,7)$ and $O_4$ realizes $\sigma_4=(13,16)(14,15)$, where we used $(\cdot,\cdot)$ to denote a transposition. Since $P$ does not mix Majorana operators in different rows, the corresponding unitary transformation $V_P$ can be implemented in depth $d_p=\mathcal{O}(l_x)$ and using $n_P=\mathcal{O}(l_x^2l_y)$ nearest neighbor two qubit gates in the horizontal JW encoding (described below). Consequently, the (asymptotic) bounds on the depth and gate-count for implementing $U_{\mathrm{pair}}$ and $U_{\mathrm{resh}}^{(r)}$ are identical in both labellings of modes. 

\subsection{Qubit labellings, encodings and transformations}

\begin{figure*}[h!]
\centering
\includegraphics[width=14cm]
{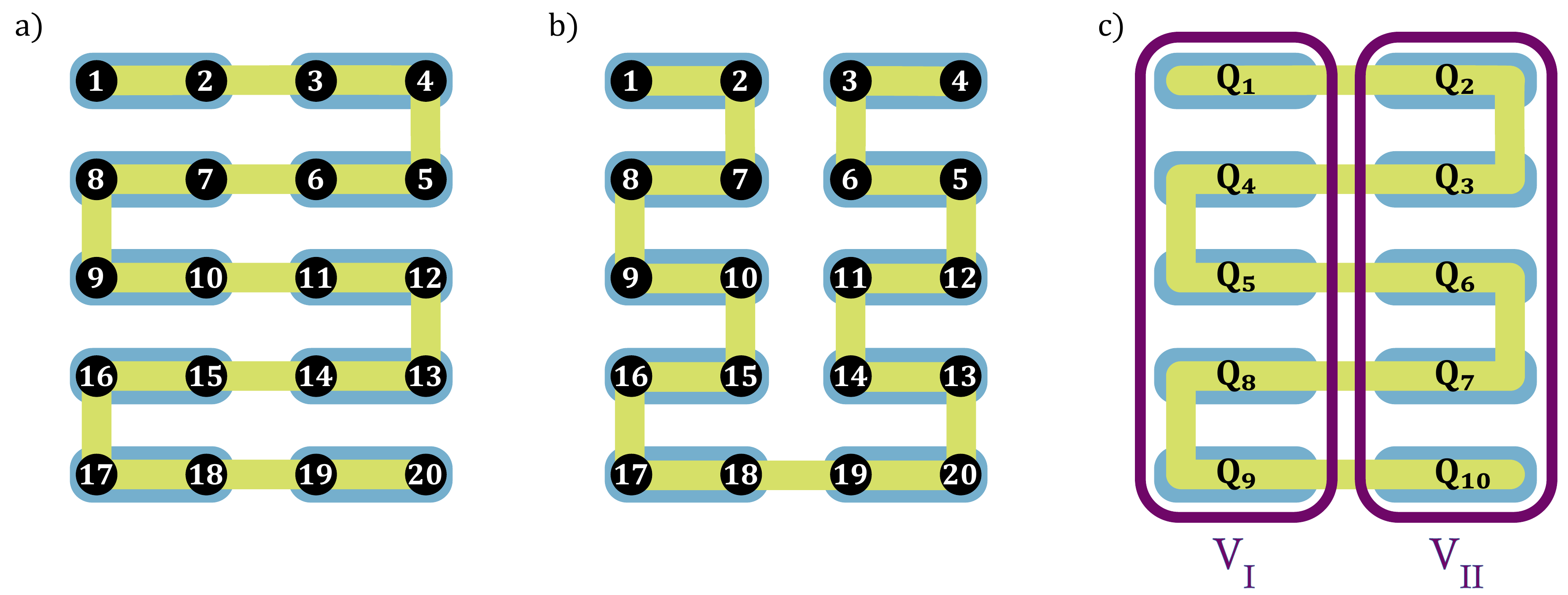}
	\caption {\textbf{Example of qubit labellings, encodings and transformations used in this section, for $2N= 2l_x l_y$ Majorana modes encoded in $l_x\times l_y$ rectangular qubit lattice for $l_x=2,l_y=5$.} \\ Dots in the figures represent Majorana modes while qubits are represented by blue rounded rectangles. Parts a) and b) present \emph{horizontal JW} and \emph{vertical JW} qubit encodings, respectively. These encodings differ by the locality structure of the encoded Majorana modes, indicated by the green "snake". In the horizontal (vertical) encoding, neighboring modes are arranged along a "snake" which traverses subsequent qubit rows horizontally (columns vertically). Part c) represents a \emph{snake} labelling of qubits ($Q_1 - Q_{10}$), as well as two \emph{JW-vertical} (see \eqref{eq:JW-verticalO} for definition) transformations $V_{I}, V_{II}$, i.e. transformations which only act on modes contained in one column of qubits.}
   \label{fig:vertical_horizontal_encoding}
\end{figure*}
In the following we will use a \emph{horizontal JW} fermion-to-qubit encoding where a rectangular $l_x \times l_y$ grid of qubits encodes $2N=2l_x l_y$ Majorana modes via the Jordan-Wigner transformation. In this encoding (see Fig. \ref{fig:vertical_horizontal_encoding}a), adjacent qubits encode adjacent pairs of Majorana modes with a locality structure consistent with the "snake" labelling in Fig. \ref{fig:2Dlabeling}. Nearest-neigbor qubit gates are those between qubits neighboring each other in either row or column of the 2D layout. In the horizontal JW encoding, any horizontal FLO transformation (recall Eq. \eqref{eq:vertical_horizontalO}) can be implemented in depth $\mathcal{O}(N^{1/2})$ and using $\mathcal{O}(N^{3/2})$ nearest neighbor two-qubit gates, each acting only on qubit pairs with both qubits in the same row of the 2D layout. However, a vertical FLO cannot trivially be implemented with the same gate complexity. For example, to implement a transformation between modes 1 and 8 in Fig. \ref{fig:vertical_horizontal_encoding}a, we need to take into account the parities incurred by traversing modes 2-7; it won't suffice to just use nearest-neighbor qubit gates in the first column of the qubit layout. Conversely, implementing a vertical FLO transformation is easy in the \emph{vertical JW} qubit encoding (see Fig. \ref{fig:vertical_horizontal_encoding}b). More generally, in that encoding it is easy to implement a JW-vertical transformation, which is defined as follows:

\begin{def*}(JW-vertical FLO transformations) Consider a snake-like labelling of $2N=2 l_x l_y$ Majorana modes in an $l_x\times l_y$ array of qubits (see Fig. \ref{fig:2Dlabeling}b, Fig. \ref{fig:vertical_horizontal_encoding}a). Let $C_i$ ($i\in[2l_x]$) denote the subset of labels of Majorana modes belonging to the i'th column of the rectangular arrangement. Furthermore, let $\mathcal{P}_\beta= C_{2\beta-1}\cup C_{2\beta}$ ($\beta\in[l_x]$) and $W_i=\mathrm{span}_\mathbb{R}\lbrace\{\ket{j}\ |\  j\in \mathcal{P}_i\rbrace\} $. We say that $V\in\mathrm{FLO}(N)$ is JW-vertical if the corresponding orthogonal transformation has the form 
\begin{equation}\label{eq:JW-verticalO}
    O= \bigoplus_{\beta=1}^{l_x} O_\beta, 
\end{equation}
where $O_\beta$ is an orthogonal transformation acting on subspace $W_\beta$.
\end{def*}
See Fig. \ref{fig:vertical_horizontal_encoding}c for a depiction of JW-vertical FLO transformations for $l_x=2,l_y=5$. Note that any transformation which is vertical in the 2D Majorana layout, as defined in Eq. \eqref{eq:vertical_horizontalO} of the main text, is also JW-vertical. The converse is not true (because a JW-vertical transformation can mix modes between two adjacent columns in the 2D Majorana layout).

In what follows we will construct a transformation $\Gamma$ which, in a system under the horizontal JW encoding, implements a JW-vertical transformation (see Fig. \ref{fig:vertical_horizontal_encoding}c) with the desired gate depth and count. This is achieved by conjugating with $\Gamma$ a circuit which would implement the desired JW-vertical transformation in the vertical JW encoding.  

\subsection{Compilation of vertical FLOs under horizontal JW encoding}
The compilation procedure presented in the proof below can likely be made more efficient by taking into account the structure of the unitaries  $U_{\mathrm{pair}}$, $U_{\mathrm{resh}}^{(r)}$, as well as a possible further improvement in the compilation of general JW-vertical FLO transformations. 

\begin{prop}\label{prop:compilation}
    Let $V\in\mathrm{FLO}(N)$ be a JW-vertical FLO transformation on a $l_x \times l_y$ array of qubits. Then $V$ can be implemented in depth $d_V=\mathcal{O}(l_x)+\mathcal{O}(l_y)$ and using $n_V =  \mathcal{O}(l_y^2 l_x)$ nearest neighbor two qubit gates, assuming one extra layer of qubits is available per row.
\end{prop}

The above proposition shows the desired depths and gate-counts necessary to implement $U_{\mathrm{pair}}$ and $U_{\mathrm{resh}}^{(r)}$ in a rectangular array of qubits. Indeed both unitaries are JW-vertical FLO transformations and for $l_x=\mathcal{O}(N^{1/2})$, $l_y=\mathcal{O}(N^{1/2})$ we have gate depth $d_V=\mathcal{O}(N^{1/2})$ and nearest neighbor gate-count $n_V= \mathcal{O}(N^{3/2})$. 

\begin{proof}[Proof of Proposition \ref{prop:compilation}]
    We will follow the proof strategy presented in \cite{Boxio18} for the case of JW-vertical \emph{passive} FLO transformations and extend it to \emph{active} ones. The idea is to map, via a suitable unitary transformation $\Gamma$, between the circuits implementing the desired JW-vertical FLO in the horizontal and vertical JW encodings -- see Fig. \ref{fig:vertical_horizontal_encoding}a and b for an example. 

   We first note that in the vertical encoding, an implementation of an arbitrary JW-vertical $V$ can be achieved in $\mathcal{O}(l_y)$ depth and by using $\mathcal{O}(l_x l_y^2)$ gates. This is because each of the $l_x$ transformations $O_\beta$ in the decomposition \eqref{eq:JW-verticalO} can be realized by a product of $\mathcal{O}(l_y^2)$ "nearest neighbor"  rotations,  $\exp(i \varphi \gamma_a \gamma_b$),  where $a,b$ are indices of neighboring Majorana modes (see e.g. \cite{Boxio18} or \cite{oszmaniec22} for the detailed proof). In the vertical encoding, such transformations are realized by single qubit Pauli $Z$ rotations, $\exp(i\varphi Z_i)$, $i\in[N]$ (they realize transformations $\exp(i\varphi \gamma_{2i-1} \gamma_{2i})$) and local gates generated by the following local couplings 
\begin{equation}
    \mathcal{C}_{\alpha \beta} = \lbrace X_\alpha X_\beta, X_\alpha Y_\beta , Y_\alpha X_\beta , Y_\alpha Y_\beta  \rbrace\ ,
\end{equation}
where $\alpha,\beta$ are neighboring qubits that belong to the same column in the qubit arrangement. Each of the transformations $O_\beta$, which act on disjoint subspaces, requires gate depth $\mathcal{O}(l_y)$ and can be realized in parallel, yielding the claimed gate depth.

In the horizontal encoding, single qubit Pauli rotations still generate $\exp(i\varphi \gamma_{2i-1} \gamma_{2i})$, while generators from $ \mathcal{C}_{\alpha \beta} $ have to be multiplied by an additional parity term $\prod_{\alpha<i<\beta} Z_i$ in order to generate the right fermionic transformation. The above discussion shows that if we can find $\Gamma$ such that for all $i\in[N]$ we have $\Gamma Z_i \Gamma^{\dagger}=Z_i$ and, furthermore, for all neighboring qubits  $\alpha,\beta$ that belong to the same column, and for all $A\in  \mathcal{C}_{\alpha \beta} $, we have
\begin{equation} \label{eq:GammaCond1}
    \Gamma A\Gamma^{\dagger} = A \prod_{\alpha<i<\beta} Z_i\ ,
\end{equation}
then we have 
\begin{equation}\label{eq:transferGamma}
    \Gamma V_v \Gamma^{\dagger} = V\ ,
\end{equation}
where $V$  is the circuit realizing in the horizontal JW encoding a JW-vertical FLO transformation, which in the vertical JW encoding is realized by $V_v$.  In what follows we explicitly construct a transformation $\Gamma$ that satisfies the above conditions and can be realized in depth $d_\Gamma = \mathcal{O}(l_x)+\mathcal{O}(l_y)$, by using $n_\Gamma = \mathcal{O}(l_x l_y) $ nearest neighbor two-qubit gates, and a single auxiliary qubit per row. This, together with condition \eqref{eq:transferGamma}, gives the claimed estimates for the depth and gate counts required to implement $V$:
\begin{equation}
    d_V\leq d_\Gamma + d_{V_v}+ d_{\Gamma^\dagger} = \mathcal{O}(l_x)+\mathcal{O}(l_y)\ ,
\end{equation}
\begin{equation}
    n_V\leq n_\Gamma + n_{V_v}+ n_{\Gamma^\dagger} = \mathcal{O}(l_y l_x) + \mathcal{O}(l_y^2 l_x) = \mathcal{O}(l_y^2 l_x) . 
\end{equation}

\begin{figure*}[h]

\centering
\includegraphics[width=18cm]
{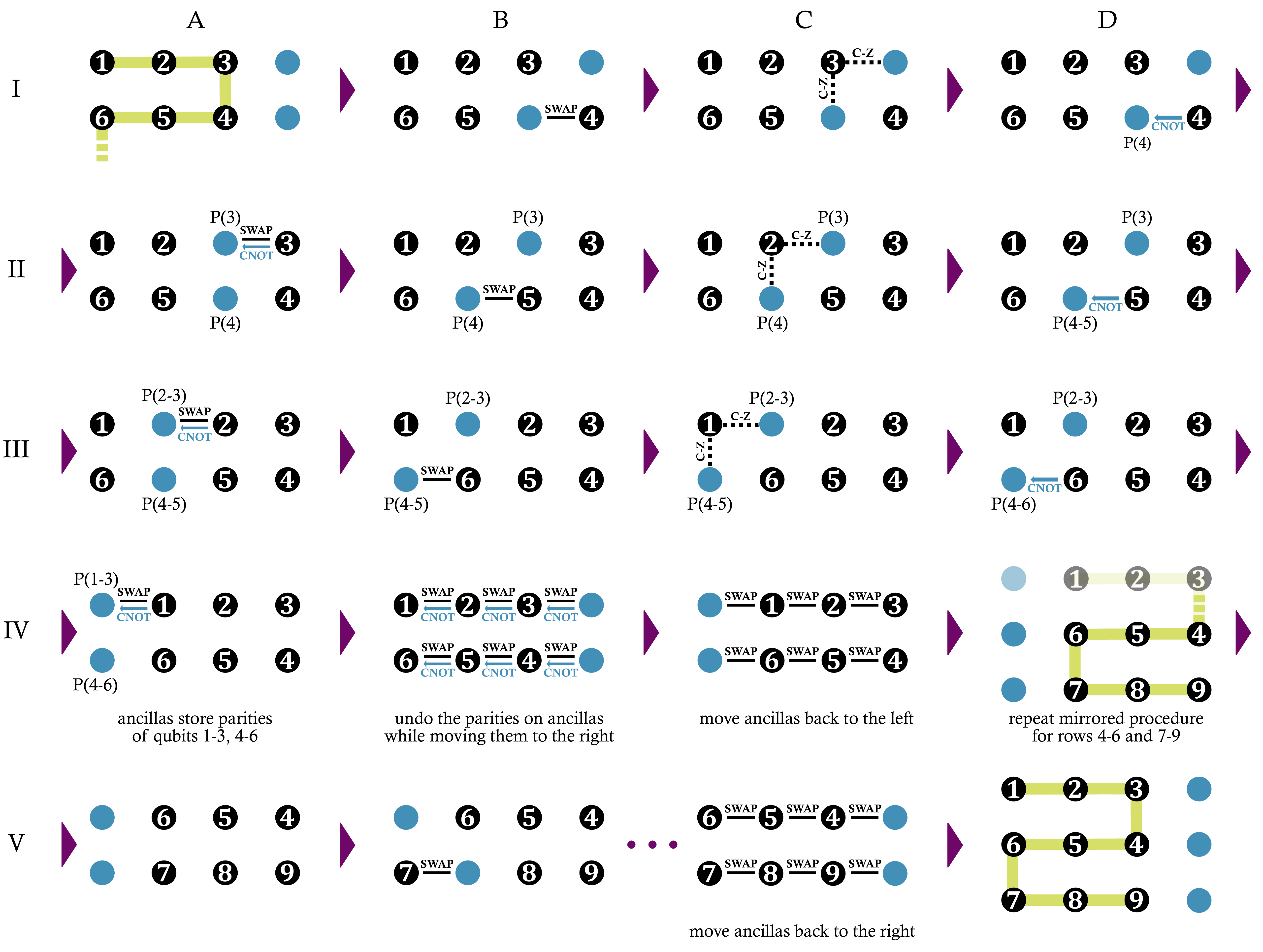}
	\caption {\textbf{A graphical presentation of a circuit implementing a unitary $\Gamma$ for $l_x=3$, $l_y=3$. } \\
 Black and blue dots represent the system and ancillary qubits respectively. The subsequent diagrams depict subsequent stages of the computation. A detailed explanation of different parts of the circuit is given in the last part of the proof of Proposition \ref{prop:compilation}.  Specific steps in the figure are referenced in text by their row (I-V) and column (A-D). The qubits $i-j$, whose parities are stored in a given ancillary qubit, are indicated by annotation P($i-j$) placed below/above the ancillary qubit.}
   \label{fig:vertical_FLO_implementation}
\end{figure*}

Before presenting a circuit that implements the desired unitary $\Gamma$, we first make some general observations which will be useful in later analysis. We first impose the condition that $\Gamma$ be diagonal in the computational basis because then $\Gamma Z_i \Gamma^\dagger = Z_i$ and it suffices to check the condition  from Eq. \eqref{eq:GammaCond1}. This condition is equivalent to requiring that for all neighboring qubits $\alpha,\beta$ belonging to the same column, for all $A\in  \mathcal{C}_{\alpha \beta} $, and for arbitrary computational basis vectors $\ket{s},\ket{s'}$ such that $\bra{s} A \ket{s'}\neq 0$, we have
\begin{equation}\label{eq:eigenvaluesGamma}
    \bar{\lambda}_s \lambda_{s'} = (-1)^{\sum_{\alpha<i<\beta} s_i}=(-1)^{\sum_{\alpha<i<\beta} s'_i}\ , 
\end{equation}
where $s=s_1s_2\ldots s_N\ ,\ s'=s'_1s'_2\ldots s'_N$ are bitstrings corresponding to states $\ket{s},\ket{s'}$, $s_i, s_i'$ describe state 0 or 1 on $i$-th qubit, and $\lambda_s,\lambda_{s'}$ are the corresponding eigenvalues of $\Gamma$, which we impose to be real. For every $\mathcal{C}_{\alpha\beta}$, the characterization of bitstrings $s,s'$ such that $\bra{s} A \ket{s'}\neq 0$ is straightforward since $A\in \mathcal{C}_{\alpha \beta} $ act as simultaneous bitflips on the computational basis states (up to an overall phase). Consequently we get that bitstrings $s,s'$ for which \eqref{eq:eigenvaluesGamma} has to be satisfied are of the following form (for any pair of qubits $\alpha,\beta$ which belong to the same column and are neighbors)
\begin{equation}\label{eq:bitstringsSUBSET}
    s_\alpha = \neg s'_\alpha\ ,\ s_\beta = \neg s'_\beta \ ,\ s_i = s'_i\ \,\,\text{for } i\notin\{\alpha,\beta\}\ .
\end{equation}

In what follows, we present a concrete construction of a circuit implementing the desired $\Gamma$ for $l_x=3$ and $l_y=3$. The steps are labelled in correspondence with Fig. \ref{fig:vertical_FLO_implementation}, where the circuit diagrams are indexed by rows (I-V) and columns (A-D). The construction sketched below is a slight modification  of the circuit presented in \cite{Boxio18} for the case of passive fermionic linear optics.

We define a \emph{right(left)-connected} qubit pair in the horizontal JW encoding -- see Fig. \ref{fig:vertical_horizontal_encoding}c -- as a pair of qubits in neigboring rows for which the snake segment connecting the qubits in the pair has a vertical part on the right (left) side of the qubits' respective rows.  E.g. in Fig. \ref{fig:2Dlabeling}a pairs (1,6), (2, 5), (3, 4) are right-connected and pairs (6, 7), (5, 8), (4, 9) are left-connected.

 \begin{itemize}
     \item[AI-AIV:] In these steps we cover the case of right-connected qubit pairs in qubit rows 1-3 and 4-6, i.e. qubit pairs (1,6), (2, 5), (3, 4). First, the ancilla qubits are moved horizontally to the left through the layout, acquiring combined parities $s_i$ of qubits $a-b$: $P(a-b)=(-1)^{\sum_{i=a}^{b} s_i}$ to the right of them. Then, via the suitably orchestrated application of C-Z gates between system qubits and ancillary qubits, the right relative phases (eigenvalues) are imprinted on the computational basis states relevant for specific neighboring qubits $C_{\alpha\beta}$ (i.e. the states satisfying conditions in \eqref{eq:bitstringsSUBSET}).
     \begin{itemize}
         \item[AI:]  We use auxiliary qubits that are initialized in the state $\ket{0}$ and placed to the right of the system qubits.
         \item[BI:] Ancilla qubit in row 2 is moved to the left, by application of a SWAP gate.
         \item[CI:] C-Z gates are applied between: (qubit 3, ancilla in row 1) and (qubit 3, ancilla in row 2). They are responsible for implementing the correct relative phases between states $\ket{s},\ket{s'}$ satisfying \eqref{eq:bitstringsSUBSET} for qubit pair (3, 4).
         \item[DI:] Ancilla qubit in row 2 acquires parity of qubit 4, by application of CNOT gate.
         \item[AII:] Ancilla qubit in row 1 is moved to the left by application of SWAP gate and acquires parity of qubit 3, by application of CNOT gate.
         \item[BII-AIII:] Steps analogous to BI-AII are repeated, moving ancillas further to the left. The C-Z gates applied in CII are responsible for implementing the correct relative phases between states $\ket{s},\ket{s'}$ satisfying \eqref{eq:bitstringsSUBSET} for qubit pair (2, 5). Ancilla in row 1 now stores parities of qubits 2-3, ancilla in row 2 now stores parities of qubits 4-5.
          \item[BIII-AIV:] Steps analogous to BI-AII are repeated, moving ancillas further to the left.  The C-Z gates applied in CIII are responsible for implementing the correct relative phases between states $\ket{s},\ket{s'}$ satisfying \eqref{eq:bitstringsSUBSET} for qubit pair (1, 6). Ancilla in row 1 now stores parities of qubits 1-3, ancilla in row 2 now stores parities of qubits 4-6. All auxiliary qubits are located to the left of the system qubits, all the phases of eigenstates $\ket{s},\ket{s'}$ relevant for right-closed neighboring qubits satisfy \eqref{eq:eigenvaluesGamma}.
         \item[BIV:] Ancilla qubits are moved back to the right while their internal states are being uncomputed by application of C-NOTs and SWAPs applied in reverse order, relative to steps BI - AIV .
         \item[CIV:] Ancilla qubits are moved to the left (without any coupling to the system qubits), by application of SWAPs.
     \end{itemize}
     \item[DIV, AV, BV, ..., DV: ] In these steps we repeat an analogous procedure in order to cover left-connected qubit pairs ((6, 7), (5, 8), (4,9)). The whole process is repeated by mirroring the procedure previously applied in steps AI-AIV to rows 4-6, 7-9. This ensures that the relative phases $\lambda_s,\lambda_{s'}$  between eigenvalues of eigenstates $\ket{s},\ket{s'}$ relevant to connections 6-7, 5-8, 4-9 (in the sense of Eq. \eqref{eq:bitstringsSUBSET}) are correctly assigned.
     \begin{itemize}
        \item[AV: ] We begin with a mirrored setup from step AI. Auxiliary qubits are initialized in the state $\ket{0}$ and placed to the left of the system qubits.
        \item[BV:] Analogously to step BI, ancilla qubit in row 4 is moved to the right, by application of a SWAP gate.
        \item[...:] Subsequent steps analogous to CI-BIV are performed.
         \item[CV:] Ancilla qubits are moved to the right (without any coupling to the system qubits), by application of SWAPs. Each ancilla ends up in its initial state, $\ket{0}$.
         \item[DV:]  We end up with the initial setup, having effectively effectively realized a real, diagonal, and unitary transformation $\Gamma$ on the system qubits.
     \end{itemize}
 \end{itemize}

The above procedure can be parallelized by operating simultaneously on disjoint pairs of rows in both stages of the computation. Furthermore, the same procedure can be easily generalized to arbitrary $l_x$ and $l_y$. After a simple count, we obtain the claimed bounds on the depth and two-qubit gate counts, namely $d_\Gamma = \mathcal{O}(l_x)+\mathcal{O}(l_y)$ and $n_\Gamma = \mathcal{O}(l_x l_y)$.
\end{proof}
\end{document}